\pdfoutput=1
\documentclass[11pt,a4paper]{article}

\usepackage[utf8]{inputenc}
\usepackage[T1]{fontenc}
\usepackage{amsmath,amssymb,amsthm}
\usepackage{graphicx}
\usepackage{booktabs}
\usepackage{array}
\usepackage{fancyhdr}
\usepackage{natbib}
\usepackage{algorithm}
\usepackage{algpseudocode}
\usepackage[colorlinks=true,linkcolor=blue,citecolor=blue,urlcolor=blue]{hyperref}
\usepackage[capitalise,noabbrev]{cleveref}
\usepackage{arxiv}
\graphicspath{{figures/}{./}}

\newtheorem{proposition}{Proposition}[section]
\newtheorem{remark}{Remark}[section]

\DeclareMathOperator{\sgn}{sgn}
\DeclareMathOperator{\Cov}{Cov}

\title{Matched generating elements in maximum entropy density reconstruction\thanks{Preprint, 2026.}}
\author{
  \textbf{Serhii Zabolotnii}$^{1,2,3}$ \\[0.2em]
  \small $^{1}$Cherkasy State Business College, 243 Chornovola St., Cherkasy 18028, Ukraine \\
  \small $^{2}$State Scientific Research Institute of Armament and Military Equipment Testing and Certification, Cherkasy, Ukraine \\
  \small $^{3}$Uzhhorod National University, Uzhhorod, Ukraine \\[0.2em]
  \small ORCID: \href{https://orcid.org/0000-0003-0242-2234}{https://orcid.org/0000-0003-0242-2234} \\
  \small Corresponding author: \texttt{zabolotnii.serhii@csbc.edu.ua}
}
\date{}

\begin{document}

\begin{NoHyper}
\maketitle
\end{NoHyper}

\begin{abstract}
Moment-constrained maximum entropy (MaxEnt) reconstructs a density as an exponential family whose sufficient statistics are chosen constraint functions. We study this choice separately from the numerical coordinates used to solve the dual problem. An odd-only family on a symmetric support forces $f(x)f(-x)$ to be constant, so parity matching is necessary. A logarithmic--rational constraint produces the Student/Cauchy family and remains population-defined when required integer moments do not exist; a parity-matched one-parameter path provides a reproducible fractional-exponent alternative containing the monomials. The numerical study uses a common damped-Newton solver and audits the six-distribution benchmark of Rajan et al. Stable parity-block divided differences, residual replay in the original constraint coordinates, and independent adaptive integration identify a raw-coordinate false convergence on D5. Under the corrected computation, the selected PATP arm is within $1.22$ times the oracle on D5 but retains selected-to-oracle error ratios of $30.7$, $8.88$, and $3.33$ on D2, D4, and D6 at $n=12$. These checks preserve the negative selector conclusion without treating a numerical artifact as statistical evidence. The Pearson row is reproduced, whereas the published GOPoly rows are not under the printed protocol. Matched LogRat results show family capacity; neither tested selector supports automatic choice among constraint families.
\end{abstract}



\medskip
\noindent\textbf{Nomenclature.}
{\small
$f(x)$ probability density; $\mathcal{S}(f)$ Shannon entropy; $\varphi_i$ generating-element (constraint/basis) function; $\mu_i$ constraint target; $\lambda$ Lagrange multipliers; $Z(\lambda)$ partition function; $\Gamma$ dual potential; $H$, $\kappa_H$ dual Hessian and its condition number; $D=[-L,L]$ truncated support; $\alpha$ PATP transition parameter; $p_i(\alpha)$ PATP exponent map; $\psi_X(u)$ characteristic function; $p$, $S$ T-MaxEnt frequency scale and harmonic count; $\epsilon$ ridge parameter; $N$ sample size; $\hat q_\beta$ estimated $\beta$-quantile.
\par}

\section{Introduction}
\label{sec:intro}

Reconstructing a probability density from a finite set of generalized moments is a classical statistical problem. Its foundation is the moment problem \citep{shohat1943problem,akhiezer1965classical}: a finite moment sequence does not determine a density, so a reconstruction rule must add a principle that selects one member of the feasible set. Maximum entropy \citep{jaynes1957information,kapur1989maximum} is the standard such principle, and its solution is the exponential family whose sufficient statistics are the constraint functions \citep{barndorff1978information,csiszar1975divergence}. Entropy densities of this kind are used in statistics and econometrics as flexible density approximants \citep{zellner1988calculation,wu2003calculation} and as conditional-density models for skewness and kurtosis \citep{rockinger2002entropy}. In applied uncertainty evaluation, they convert a few propagated moments into a coverage interval when the Gaussian assumption fails \citep{jcgm100,jcgm101,rajan2018moment,morello2020gum}. In every such use the constraint functions are chosen before the density is fitted, and the accuracy of the reconstruction --- especially in the tails --- is governed by that choice. Which constraint functions to use, and why, is the question this paper addresses.

Among moment-based approximants --- Pearson and Johnson translation systems \citep{johnson1949systems}, Gram--Charlier and Edgeworth expansions \citep{blinnikov1998expansions}, saddlepoint approximations \citep{daniels1954saddlepoint}, and general moment-based approximants \citep{provost2005moment} --- the moment-constrained maximum entropy (MaxEnt) method \citep{jaynes1957information,mead1984maximum} has a distinguished position: it selects the entropy maximizer consistent with the available moments, does not require a named distributional family, and guarantees a nonnegative density, which series expansions do not. These properties have made it a widely used reconstruction engine in expanded uncertainty evaluation \citep{rajan2016benchmark,rajan2018moment} and structural reliability \citep{li2011combined,zhang2013entropy}.

Classical MaxEnt, however, relies on monomial constraints $x^i$, which present two well-documented difficulties. First, \textbf{numerical scaling}: in raw monomial coordinates the basis functions become highly correlated at moderate-to-high orders, and the dual Hessian (a Hankel-type covariance matrix) can be severely ill-conditioned; stabilized reparameterizations address this solver problem without changing the constraint span \citep{bandyopadhyay2005stable,abramov2007improved}. Second, \textbf{heavy-tailed divergence}: for heavy-tailed laws the required integer moments may not exist --- for the Cauchy distribution $\mathbb{E}|X|^{r}$ is finite only for $r<1$, and for Student's $t$ with $\nu$ degrees of freedom only for $r<\nu$ --- so the population constraints of a four-moment monomial MaxEnt are undefined and the sample-based surrogate problem can be infeasible (\cref{sec:exp1}). The fractional-moment MaxEnt (FM-MEM) literature mitigates these difficulties by replacing integer exponents with fractional ones \citep{noviinverardi2003fractional,zhang2013entropy}, but it still leaves the exponent search, signed-support parity, and choice between qualitatively different constraint families unresolved.

Our starting point is the observation that the monomials are not fundamental: they are one possible vector of constraint functions, or \emph{generating element} in the Kunchenko terminology \citep{kunchenko2002polynomial,kunchenko2005approximation,kunchenko2006stochastic}. The MaxEnt density $f\propto\exp(\sum_i\lambda_i\varphi_i)$ is an ordinary exponential family with these functions as its sufficient statistics. We treat their choice as the modeling decision and distinguish it from the coordinates used by the numerical solver. The paper studies three constraint families under a common dual algorithm; its contributions are the parity-admissibility result, a tail-class design map, and a numerical audit that separates family capacity from coordinate-dependent convergence:
\begin{itemize}
  \item \textbf{Fractional-power element} (PATP, \cref{sec:patp}): a one-parameter family in which the parametrically adaptive transition polynomial \citep{zabolotnii2026parametrically} ties all exponents to a single scalar $\alpha\in[0,1]$, so exponent selection becomes a deterministic one-dimensional scan rather than the $m$-dimensional non-convex search of FM-MEM \citep{zhang2013entropy,zhang2020fractional,li2019improved}, with the classical monomials recovered exactly at $\alpha=1$. We call this element \emph{agnostic} because it assumes no tail class; it is parity-matched and therefore applicable on signed supports.
  \item \textbf{Parity-admissibility theorem} (\cref{prop:parity}): any generating element built solely from odd functions --- including the signed-parity (Form-B) basis of the original PATP estimation framework --- yields $f(x)f(-x)\equiv Z^{-2}$ and thus cannot represent any non-uniform symmetric density. Parity matching is therefore necessary for an element intended to represent non-uniform symmetric densities on a symmetric support.
  \item \textbf{Trigonometric (characteristic-function) element} (T-MaxEnt, \cref{sec:tmaxent}): an adaptation of the classical trigonometric-moment MaxEnt machinery --- long standard on the circle and in spectral estimation \citep{burg1975mesa,lang1982mem,gatto2007gvm} --- to PDF reconstruction on truncated supports, with empirical characteristic-function (ECF) plug-in constraints \citep{feuerverger1977ecf} and a CF-amplitude frequency-design rule. Its constraints exist for any probability distribution without moment assumptions, and its basis entries are bounded. The experiments examine where this representational choice is useful; bounded entries alone do not imply a bounded Hessian condition number.
  \item \textbf{Logarithmic--rational element} (\cref{sec:lograt}): the constraint $\log(1+(x/s)^2)$ makes the exponential family $f\propto(1+(x/s)^2)^\lambda$ --- exactly the Student/Cauchy family --- so it represents \emph{algebraic} (power-law) tails, which the fractional and trigonometric elements do not naturally represent --- their fitted log-densities are, respectively, fractional-polynomial and bounded in $|x|$, not logarithmic, so neither yields a power-law decay. This is a \emph{matched} element: it encodes an algebraic-tail assumption and, when that assumption holds, recovers the tail index from a single constraint.
  \item \textbf{Numerical-coordinate audit} (\cref{sec:dual,sec:headtohead}): invertible changes of constraint coordinates preserve the represented density family but not the Hessian condition number. In the Rajan benchmark, parity-block divided differences, residual replay in the original PATP coordinates, and independent adaptive integration identify a raw-coordinate false convergence that a discrete Newton residual alone does not detect.
  \item \textbf{Design map and matching principle} (\cref{sec:genelement,sec:discussion}): a controlled comparison across tail classes (Cauchy, Student-$t$, $\alpha$-stable, Gaussian) that records the empirical capacities and mismatch costs of the three families. The results support tail-informed comparison of constraint families, but they do not establish a universal or automatic routing rule.
\end{itemize}

The same generating-element formalism appears in Kunchenko's parameter-estimation and stochastic-polynomial constructions \citep{kunchenko2002polynomial,kunchenko2006stochastic}. Here it is used more narrowly for MaxEnt density reconstruction: the element supplies the sufficient statistics of an exponential family.

We evaluate the framework in four seeded studies (\cref{sec:results}), with Experiments 1--2 replicated over $20$ seeds. On Cauchy data the fractional element reconstructs the body (KS $0.068\pm0.026$) but not the tail; on a bimodal mixture one scan-selected member improves reconstruction MSE over the monomial baseline, although the same entropy selector fails on other targets; and the matched logarithmic element recovers the Cauchy tail index from one constraint. The external benchmark gives mixed evidence: Pearson is reproduced, the published GOPoly rows are not reproduced by the literal printed protocol, and truth-matched LogRat improves a heavy-tail domain-extension panel. Cross-family automatic selection is an unresolved model-selection problem; the matched LogRat results are oracle-family evidence, not an operational routing rule.

The remainder of the paper is organized as follows. \Cref{sec:related} reviews related work. \Cref{sec:math} develops the dual formulation and the three generating elements. \Cref{sec:numerical} details the implementation. \Cref{sec:results} reports the four experiments, including the Rajan benchmark audit (\cref{sec:headtohead}); \cref{sec:discussion} distills the design map, \cref{sec:limitations} states limitations, and \cref{sec:conclusion} concludes.

\section{Related Work}
\label{sec:related}

\subsection{Moment constraints, fractional orders, and relaxed feasibility}
The use of fractional moments as MaxEnt constraints originates with \citet{noviinverardi2003fractional}, who showed that a few selected fractional moments can replace many integer moments and proposed an Akaike-type selection rule. \citet{gzyl2010hausdorff} used fractional moments synthesized from an integer-moment sequence to reduce the ill-conditioning of high-order Hausdorff recovery, while \citet{cottone2009fractional} connected complex fractional moments to the characteristic function through the Mellin transform. In structural reliability, FM-MEM combines fractional exponents with dimensional reduction and treats the exponents and Lagrange multipliers as a non-convex optimization problem \citep{zhang2013entropy}. Later variants use Laplace transforms, stratified moment estimation, genetic search, polynomial surrogates, and iterative integration \citep{li2019improved,xu2019novel,zhang2020fractional,li2024improved,wang2025iterative}; multimodal extensions are also available \citep{li2022multimodal}. In a different sampling problem, \citet{aleagobe2023length} compared maximum entropy and empirical likelihood under moment conditions for length-biased data.

Moment relaxation addresses another failure mode. \citet{nghiem2026relaxed} formulated continuous-domain density estimation with relaxed moment constraints as a convex Kullback--Leibler problem and derived a finite-dimensional dual. Relaxation can absorb sampling error or inconsistent targets; it does not choose the constraint map itself. Our study keeps exact fitted constraints, varies the map, and reports when empirical targets fall outside the realizable set. PATP-MaxEnt likewise does not claim to introduce optimized fractional orders. It ties all exponents to one scalar $\alpha$, making the search reproducible and one-dimensional, and uses a parity-matched signed basis rather than the positive-support powers common in FM-MEM.

\subsection{Selecting constraint functions and numerical coordinates}
The closest data-adaptive approach is MESSY estimation \citep{tohme2024messy}. It uses symbolic regression to search for smooth functions in the MaxEnt exponent and screens candidates for reconstruction accuracy and conditioning. Our objective is different: we compare a small collection of prespecified, interpretable constraint families and ask which distribution classes they can represent. A direct numerical comparison would require matching MESSY's symbolic function library and search budget to this prespecified-family study; we therefore treat it as a conceptual comparator rather than an empirical baseline. The parity obstruction is a representational statement, whereas the stable algorithms of \citet{bandyopadhyay2005stable} and \citet{abramov2007improved} address coordinates used to solve a fixed moment problem. Because an invertible coordinate change can alter the Hessian condition number without changing the fitted density, we replay residuals in the original constraints and do not rank families by raw condition numbers.

\begin{table}[t]
  \centering
  \caption{Closest MaxEnt approaches and the design layer addressed by each.}
  \label{tab:related-comparison}
  \footnotesize
  \setlength{\tabcolsep}{3pt}
  \begin{tabular}{>{\raggedright\arraybackslash}p{0.20\linewidth}>{\raggedright\arraybackslash}p{0.29\linewidth}>{\raggedright\arraybackslash}p{0.43\linewidth}}
    \toprule
    Approach & Constraint treatment & Boundary relative to this paper \\
    \midrule
    Fractional-moment MaxEnt \citep{noviinverardi2003fractional,zhang2013entropy}
      & Selects or optimizes fractional orders
      & Primarily positive-support powers; no general parity condition or cross-family tail map. \\
    MESSY \citep{tohme2024messy}
      & Searches symbolically over smooth exponent functions
      & Data-adaptive basis discovery; no prespecified tail-family interpretation or deployable comparison with the three families studied here. \\
    Relaxed MaxEnt \citep{nghiem2026relaxed}
      & Relaxes noisy or inconsistent moment targets
      & Changes constraint tolerance rather than the sufficient-statistic map. \\
    Weak Fourier constraints \citep{soize2023fourier}
      & Enforces characteristic-function information in a weak Kullback--Leibler formulation
      & Targets learning from realizations, including high-dimensional settings; not a finite harmonic reconstruction with an ECF frequency rule. \\
    Present work
      & Compares fractional, harmonic, and logarithmic--rational maps
      & Provides parity and tail-representation results plus numerical replay; cross-family automatic selection remains unresolved. \\
    \bottomrule
  \end{tabular}
\end{table}

\Cref{tab:related-comparison} separates changes to the constraint map from changes to target tolerance or numerical coordinates. None of the listed approaches combines the parity restriction, the tail-family representation result, and original-coordinate numerical replay studied here.

\subsection{Fourier constraints and tail models}
MaxEnt under trigonometric constraints is classical in spectral estimation \citep{burg1975mesa,lang1982mem}, image reconstruction from incomplete Fourier data \citep{gull1978image}, and circular statistics \citep{gatto2007gvm}. On bounded intervals, \citet{barron1991exponential} established convergence rates for exponential families with trigonometric log-density expansions. More recently, \citet{soize2023fourier} imposed probability-measure constraints through a weak formulation of the Fourier transform under Kullback--Leibler minimization. T-MaxEnt is narrower: it uses a finite real--imaginary harmonic vector on a truncated univariate support, ECF plug-in targets with established sampling theory \citep{feuerverger1977ecf}, and a sample-size-dependent frequency screen. Bounded basis entries guarantee moment existence but do not guarantee a well-conditioned Hessian.

Power-law tails can instead be obtained by changing the entropy functional: Tsallis entropy yields $q$-exponential maximizers \citep{tsallis1988possible}. Our logarithmic--rational element keeps Shannon entropy fixed and puts the algebraic-tail assumption into the constraint. Generalized maximum entropy \citep{golan1996maximum}, boundary-corrected KDE \citep{jones1993simple}, and positive orthogonal-series estimators \citep{efromovich2010orthogonal} provide broader alternatives. We do not claim superiority over these estimator classes; a tuned comparison under matched information remains outside the present study.

\subsection{Generating-element interpretation and scope}
\label{sec:related_genelement}
Kunchenko's decomposition-space terminology calls the chosen function vector a \emph{generating element} \citep{kunchenko2002polynomial,kunchenko2005approximation,kunchenko2006stochastic}. In the original polynomial-maximization setting it supports parameter estimation; here it names the sufficient statistics of a standard MaxEnt exponential family. The existence, uniqueness, and convex-duality results remain the usual ones for moment-constrained MaxEnt and information projection \citep{borwein1991convergence,csiszar1975divergence}. Thus ``generating element'' can be read throughout as ``chosen vector of constraint functions''; the paper's new questions concern parity, tail representation, and the numerical consequences of that choice.

\section{Mathematical Formulation}
\label{sec:math}

\subsection{Dual formulation of constrained maximum entropy}
\label{sec:dual}
Let $X$ be a random variable defined on a support $D \subseteq \mathbb{R}$ with PDF $f(x)$. The goal is to maximize the Shannon entropy
\begin{equation}
  \mathcal{S}(f) = - \int_{D} f(x) \ln f(x) \, dx,
  \label{eq:entropy}
\end{equation}
subject to the normalization constraint $\int_{D} f(x) \, dx = 1$ and a set of generalized moment constraints defined by a \emph{generating element} $\Phi = \{\varphi_i\}_{i=1}^n$ --- the family of functions that spans the decomposition space in the Kunchenko sense \citep{kunchenko2005approximation,kunchenko2006stochastic}:
\begin{equation}
  \int_{D} \varphi_i(x) f(x) \, dx = \mu_i, \quad i = 1, \dots, n.
  \label{eq:constraints}
\end{equation}
The monomials $\varphi_i(x)=x^i$ are the classical choice; \cref{sec:patp,sec:tmaxent,sec:lograt} develop three others. Everything in the present subsection is independent of the choice of element.
Using the method of Lagrange multipliers, the resulting PDF has the exponential-family form
\begin{equation}
  f(x; \lambda) = \frac{1}{Z(\lambda)} \exp\left( \sum_{i=1}^n \lambda_i \varphi_i(x) \right),
  \label{eq:pdf_form}
\end{equation}
where $Z(\lambda) = \int_D \exp\left(\sum_k \lambda_k \varphi_k(x)\right) dx$ is the partition function. The dual potential to be minimized is
\begin{equation}
  \Gamma(\lambda) = \ln Z(\lambda) - \sum_{i=1}^n \lambda_i \mu_i,
  \label{eq:dual_pot}
\end{equation}
with gradient and Hessian
\begin{equation}
  \frac{\partial \Gamma}{\partial \lambda_i} = \mathbb{E}_f[\varphi_i(X)] - \mu_i,
  \qquad
  H_{ij} = \frac{\partial^2 \Gamma}{\partial \lambda_i \partial \lambda_j} = \Cov_f\bigl(\varphi_i(X), \varphi_j(X)\bigr).
  \label{eq:grad_hess}
\end{equation}
Since the Hessian is a covariance matrix, it is symmetric and positive semi-definite, so the dual problem is convex. Convexity alone, however, does not guarantee a solution: existence and uniqueness of the minimizer additionally require $Z(\lambda)<\infty$ on a neighborhood of the iterates, linearly independent constraint functions, and a target vector $\mu$ in the interior of the realizable moment set \citep{borwein1991convergence,mead1984maximum}. Each of these conditions fails somewhere in our experiments, and we report the failures explicitly (\cref{sec:exp1}).

Conditioning is not an invariant property of the represented density family. If an invertible coordinate change defines $\psi=A\varphi$, then its Hessian satisfies
\begin{equation}
  H_{\psi}=A H_{\varphi} A^{\mathsf T}.
  \label{eq:hessian_congruence}
\end{equation}
This congruence preserves positive definiteness but not eigenvalues or condition number; simple rescaling can change $\kappa(H)$ arbitrarily while leaving the fitted density unchanged after the inverse multiplier transformation. We therefore use condition numbers only as diagnostics for a specified basis scaling and solver, never as evidence that one generating-element family is intrinsically better conditioned than another.

One condition deserves emphasis for signed bases. If the exponent $\sum_i \lambda_i\varphi_i(x)$ is an odd function --- as for the basis in \cref{eq:patp_basis_pm} below with only odd-indexed terms active --- then one of its tails grows without bound, and $Z(\lambda)=\infty$ on $\mathbb{R}$ for every $\lambda\neq 0$. MaxEnt with such bases is therefore well-posed only on a bounded support $D=[-L,L]$, and the truncation length $L$ is part of the model specification, not a mere numerical grid choice. All experiments below state $L$ explicitly.

\subsection{The fractional-power element (PATP): an agnostic one-parameter family}
\label{sec:patp}
The first generating element is a one-parameter family of parity-matched fractional powers. It is the \emph{agnostic} choice --- it encodes no assumption about the target's tail class --- and it contains the classical monomials as a special case. The PATP framework \citep{zabolotnii2026parametrically} defines a quadratic exponent map controlled by a transition parameter $\alpha \in [0, 1]$:
\begin{equation}
  p_i(\alpha) = \frac{1}{i} + \left(4 - i - \frac{3}{i}\right)\alpha + \left(2i - 4 + \frac{2}{i}\right)\alpha^2,
  \label{eq:patp_exponent}
\end{equation}
which satisfies $p_i(0) = 1/i$ (fractional regime), $p_i(0.5) = 1$ for all $i$, and $p_i(1) = i$ (integer regime). In the original estimation-theoretic setting the map is instantiated through the signed-parity (Form-B) element $\sgn(x)|x|^{p_i(\alpha)}$. For density reconstruction, however, an all-odd element is structurally inadmissible --- and the obstruction applies to \emph{every} generating element, the PATP one included:

\begin{proposition}[Parity admissibility]
\label{prop:parity}
Let every function $\varphi_i$ of the generating element in \cref{eq:pdf_form} be odd on a symmetric support $D=[-L,L]$. Then the fitted density satisfies $f(x;\lambda)\,f(-x;\lambda) = Z(\lambda)^{-2}$ for all $x$; in particular, the only symmetric density attainable is the uniform density on $D$.
\end{proposition}
\begin{proof}
$\ln f(x) + \ln f(-x) = -2\ln Z + \sum_i \lambda_i\bigl(\varphi_i(x) + \varphi_i(-x)\bigr) = -2\ln Z$ since each $\varphi_i$ is odd. If additionally $f(x)=f(-x)$, then $f(x)^2 = Z^{-2}$, so $f$ is constant.
\end{proof}

The log-density of an all-odd element has no even component, so such an element cannot constrain dispersion- or kurtosis-type behavior at all; a usable element must contain even functions. We therefore define the fractional-power (PATP) element through the \emph{parity-matched} form
\begin{equation}
  \varphi_1(x; \alpha) = x, \qquad
  \varphi_i(x; \alpha) =
  \begin{cases}
    |x|^{p_i(\alpha)}, & i \text{ even},\\[2pt]
    \sgn(x)\,|x|^{p_i(\alpha)}, & i \text{ odd},
  \end{cases}
  \quad i = 2, \dots, n,
  \label{eq:patp_basis_pm}
\end{equation}
which preserves the PATP exponent map while restoring both parities. We refer to PATP-MaxEnt built on this parity-matched basis as \emph{PM-PATP} throughout. At $\alpha = 1$ the basis recovers the classical monomials exactly for every $i$ (since $|x|^{i} = x^{i}$ for even $i$ and $\sgn(x)|x|^{i} = x^{i}$ for odd $i$), so the classical MaxEnt is the $\alpha=1$ member of the family and all comparisons against it are internally consistent.

\begin{remark}[Degeneracy at $\alpha = 0.5$]
\label{rem:degenerate}
Since $p_i(0.5)=1$ for all $i$, the basis at $\alpha=0.5$ collapses to copies of $x$ and $|x|$ (design-matrix rank $2$); the corresponding fit is a two-constraint exponential tilt whose Hessian is singular up to the ridge term. This point is flagged as degenerate in all reported sweeps.
\end{remark}

\paragraph{Moment existence.} For the standard Cauchy law $\mathbb{E}|X|^{q}<\infty$ if and only if $q<1$. All PATP exponents satisfy $p_i(\alpha)<1$ for $i\ge 2$ exactly when $\alpha<0.5$, so the population fractional moments of the constraints $\varphi_i$, $i \ge 2$, exist precisely in that range. The always-included linear constraint $\varphi_1(x)=x$ has \emph{no} finite population expectation for the Cauchy law; in the experiments its target (like all others) is an empirical average over the full sample (including any draws outside $[-L,L]$), while the fitted density lives on the truncated support, so the reconstruction must be interpreted as a sample-conditional surrogate fit (\cref{sec:exp1}).

\paragraph{Selection of $\alpha$.} For each candidate $\alpha$ on the uniform grid $\{0, 0.1, \dots, 1.0\}$ (the degenerate point $\alpha=0.5$ of \cref{rem:degenerate} is excluded), the constraint targets $\mu_i(\alpha) = \frac{1}{N}\sum_{j} \varphi_i(x_j;\alpha)$ are recomputed from the \emph{full} sample --- every draw $x_j$ contributes, including any that fall outside the truncation interval $[-L,L]$ (this matters in \cref{sec:exp1}, where retained extreme draws are precisely what can push a target outside the attainable range and render the monomial problem infeasible). The dual problem is then solved by Newton's method (\cref{alg:newton}), and the parameter is selected by
\begin{equation}
  \alpha^{\ast} = \arg\min_{\alpha} \Gamma\bigl(\lambda^{\ast}(\alpha); \alpha\bigr).
  \label{eq:opt_alpha}
\end{equation}
At the converged solution $\lambda^{\ast}$, where the constraints $\mathbb{E}_f[\varphi_i] = \mu_i$ are met, $\Gamma(\lambda^{\ast}(\alpha);\alpha)$ equals the Shannon entropy of the fitted density; away from convergence the two are related by $\Gamma = \mathcal{S}(f) + \sum_i \lambda_i\bigl(\mathbb{E}_f[\varphi_i] - \mu_i\bigr)$, the extra term being $\lambda$ contracted with the constraint residual. We apply \cref{eq:opt_alpha} only to converged candidates on a common support, but it remains a model-selection heuristic, not a consequence of MaxEnt and not a validated data-driven selector. In the Rajan audit (\cref{sec:headtohead}), its selected-to-oracle error ratios at $n=12$ are $30.7$, $8.88$, and $3.33$ on D2, D4, and D6, respectively. On D5, by contrast, the ratio is $1.22$ ($0.000428\%$ selected versus $0.000350\%$ oracle). We therefore report selected and oracle-family results separately. The truncation length $L$ is a user-specified modeling assumption, fixed a priori from the expected range of the measurand (here $L=50$ for the Cauchy benchmark and $L=5$ for the mixture, both set before fitting); we quantify its effect in the supplementary material and discuss the resulting $L$-dependence in \cref{sec:limitations}.

\subsection{The trigonometric (characteristic-function) element: bounded constraints for multimodal targets}
\label{sec:tmaxent}
The second generating element is trigonometric, derived from the characteristic function $\psi_X(u) = \mathbb{E}\bigl[e^{iuX}\bigr]$, following the characteristic-function constructions of the Kunchenko school \citep{zabolotnii2026variance}. It is a natural candidate for bounded, multimodal, light-tailed targets. The basis functions are
\begin{equation}
  \varphi_{2r-1}(x) = \cos(r p x), \qquad \varphi_{2r}(x) = \sin(r p x), \qquad r = 1, \dots, S,
  \label{eq:trig_basis}
\end{equation}
where $p$ is a frequency scaling parameter, and the target moments are the real and imaginary parts of the CF at the harmonic frequencies,
\begin{equation}
  \mu_{2r-1} = \operatorname{Re}\{\psi_X(r p)\}, \qquad \mu_{2r} = \operatorname{Im}\{\psi_X(r p)\}, \qquad r = 1,\dots,S.
  \label{eq:trig_moments}
\end{equation}
In practice the targets are estimated by the ECF, $\hat\psi_N(u) = \frac{1}{N}\sum_j e^{iux_j}$, whose uniform almost-sure convergence and asymptotic normality are classical \citep{feuerverger1977ecf}; the plug-in constraints therefore have bounded influence and known sampling theory, in contrast to high-order sample moments. Because trigonometric functions are uniformly bounded between $-1$ and $1$, every Hessian entry in \cref{eq:grad_hess} is bounded in magnitude. This does not bound its condition number: nearly collinear harmonics remain possible, and \cref{eq:hessian_congruence} precludes cross-coordinate rankings from raw $\kappa_H$. Two properties still matter for heavy-tailed targets: the constraints \cref{eq:trig_moments} exist for \emph{every} distribution (no moment conditions), and the ECF noise floor is $O(1/\sqrt{N})$ uniformly in $u$.

\paragraph{Frequency design rule.} The harmonics must carry information: if $|\psi_X(Sp)|$ is at or below the ECF noise floor ($\approx 1/\sqrt{N}$), the highest-order constraints degenerate into numerical noise. We turn this into an automated, sample-size-aware rule in two steps. First, \emph{admissibility}: a configuration $(p,S)$ is admitted only if every harmonic clears the noise floor by a confidence margin, $|\hat\psi_N(jp)| \ge 3/\sqrt{N}$ for $j=1,\dots,S$ (at $N=1000$ this threshold is $0.095$; a fixed $0.05$ threshold would be its $\approx 1.6\sigma$ special case). Second, \emph{selection}: because clearing the floor alone is not enough (a coarse fundamental can clear it and still fit poorly), we choose among admissible configurations the one minimizing a held-out predictive log-score, the same cross-validated criterion used in the supplementary selection study. \Cref{sec:exp2} shows this automated rule roughly halves the accuracy penalty of the fixed-threshold default while needing no hand-set amplitude.

\subsection{The logarithmic--rational element: a matched element for algebraic tails}
\label{sec:lograt}
The two elements above are \emph{agnostic}: neither assumes a tail class, and (as \cref{sec:exp1} shows) neither can reproduce a power-law tail. The reason is structural. The MaxEnt density is $f\propto\exp(\sum_i\lambda_i\varphi_i)$, so to obtain an algebraic tail $f(x)\sim |x|^{-\beta}$ the log-density must contain a $-\beta\log|x|$ term --- but a finite combination of fixed powers $|x|^{p_i}$ produces stretched-exponential, not algebraic, tails, and a bounded trigonometric polynomial cannot grow at all; neither reproduces the logarithmic growth of $\log|x|$ over an unbounded range. The remedy is to put the logarithm into the element itself. We take the single \emph{logarithmic--rational} constraint
\begin{equation}
  \varphi^{\mathrm{LR}}_s(x) = \log\!\bigl(1 + (x/s)^2\bigr),
  \label{eq:lograt}
\end{equation}
for which the MaxEnt density is exactly
\begin{equation}
  f(x;\lambda) = \frac{1}{Z}\bigl(1+(x/s)^2\bigr)^{\lambda},
  \label{eq:lograt_density}
\end{equation}
i.e.\ the Student/Cauchy family: $\lambda=-1$ is the standard Cauchy (scale $s$), and a general $\lambda$ gives a Student-$t$ shape with tail exponent $f(x)\sim|x|^{2\lambda}$. Normalizability fixes the admissible range of $\lambda$ and depends on the support: on the whole line $D=\mathbb{R}$, $Z<\infty$ requires $\lambda<-\tfrac12$ (so that $|x|^{2\lambda}$ is integrable at infinity), whereas on a bounded support $D=[-L,L]$ the integral is finite for every $\lambda$ and the constraint $L$ again enters the model specification, exactly as for the signed bases of \cref{sec:dual}. The element is even (parity-admissible by \cref{prop:parity}), and its generalized moment $\mathbb{E}[\log(1+(X/s)^2)]$ is \emph{finite for the Cauchy law} even though no integer or fractional power moment of order $\ge 1$ is --- e.g.\ $\mathbb{E}_{\mathrm{Cauchy}}[\log(1+X^2)] = 2\ln 2$. Under the matched model itself the same moment has the closed form $\mathbb{E}_f[\log(1+(X/s)^2)] = \psi(-\lambda)-\psi(-\lambda-\tfrac12)$ (digamma $\psi$; consistent with $2\ln 2$ at $\lambda=-1$), so a single-constraint fit on the whole line $\mathbb{R}$ reduces to a one-dimensional root-find with no quadrature or truncation (\cref{sec:genelement}). Several scales $s$ may be stacked, $\{\varphi^{\mathrm{LR}}_{s_k}\}$, giving $f\propto\prod_k(1+(x/s_k)^2)^{\lambda_k}$, a rational family that interpolates tail index and body width. This is a \emph{matched} element: it presupposes an algebraic tail. \Cref{sec:genelement} shows it recovers the Cauchy tail index from a single constraint, and \cref{sec:discussion,sec:limitations} record its converse cost --- on a light-tailed (Gaussian) target it is the wrong element and underperforms the agnostic ones. We contrast it with a \emph{bounded} rational element $1/(1+(x/s)^2)$, whose exponent vanishes as $|x|\to\infty$ and which therefore flattens rather than reproduces the tail.

\section{Numerical Implementation}
\label{sec:numerical}

\subsection{Damped Newton solver for the dual problem}
\label{sec:newton}
We minimize \cref{eq:dual_pot} with the damped Newton iteration of \cref{alg:newton}. All integrals are evaluated by the rectangle rule on a fixed grid (resolutions in \cref{sec:results}). A ridge term $\epsilon I$ is added to the Hessian before solving the Newton system, to keep it invertible under near-collinear constraints; \cref{sec:exp2} verifies that the reported results are insensitive to $\epsilon$ over four orders of magnitude. The backtracking line search accepts a step when the dual potential does not increase by more than a slack tolerance ($10^{-4}$, which absorbs the rectangle-rule quadrature error in $\Gamma$); if all ten trials fail, the solver takes a single damped step of length $0.1$ as a heuristic continuation and proceeds. This fallback cannot invalidate any reported result: convergence is certified \emph{only} by the gradient-norm test (line~\ref{lin:converge} of \cref{alg:newton}), which is evaluated on the iterate itself and is independent of how that iterate was reached, so an uphill fallback step can at worst lead to a \emph{reported non-convergence}, never to a false ``converged'' verdict or an accepted infeasible fit. Convergence from arbitrary starting points is accordingly \emph{not} guaranteed; the solver reports failure when the unnormalized density overflows, the partition function vanishes, the Newton system is singular, or the gradient tolerance is not met within the iteration budget, and every such outcome is reported explicitly in the experiments (a more conservative alternative would replace the fallback with an immediate Armijo failure return, at the cost of a few additional reported non-convergences).

\begin{algorithm}[t]
\caption{Damped Newton iteration for the MaxEnt dual problem}
\label{alg:newton}
\begin{algorithmic}[1]
\Require basis matrix $\Phi \in \mathbb{R}^{G\times n}$ on the grid, targets $\mu$, grid step $\Delta x$, ridge $\epsilon = 10^{-8}$, tolerance $10^{-6}$, $\text{maxit} = 100$
\State $\lambda \gets 0$
\For{$k = 1, \dots, \text{maxit}$}
  \State $w \gets \exp(\Phi\lambda)$; \textbf{fail} if any entry overflows \label{lin:overflow}
  \State $Z \gets \mathbf{1}^{\top} w \,\Delta x$; \textbf{fail} if $Z = 0$; \quad $f \gets w / Z$
  \State $g \gets \Phi^{\top} f \,\Delta x - \mu$ \Comment{dual gradient, \cref{eq:grad_hess}}
  \If{$\lVert g \rVert_2 < \text{tol}$} \Return $\lambda$ \Comment{converged} \label{lin:converge}
  \EndIf
  \State $H \gets \Phi^{\top}(\Phi \circ f)\,\Delta x - m m^{\top}$ \Comment{$m = \Phi^{\top} f \Delta x$; unregularized Hessian}
  \State solve $(H + \epsilon I)\, d = g$ (QR fallback); \textbf{fail} if singular
  \State backtracking: find largest $t \in \{1, \tfrac12, \dots, 2^{-9}\}$ with $\Gamma(\lambda - t d) \le \Gamma(\lambda) + 10^{-4}$
  \State $\lambda \gets \lambda - t d$ \quad (if no $t$ accepted: $\lambda \gets \lambda - 0.1\, d$)
\EndFor
\State \Return last iterate, flagged as not converged
\end{algorithmic}
\end{algorithm}

Throughout the paper, any reported $\kappa_H$ is the exact 2-norm condition number of the \emph{ridge-regularized} Hessian $H + \epsilon I$ in the stated coordinates at the final iterate. It is a solver diagnostic only (\cref{eq:hessian_congruence}); without the ridge, the Hessian at $\alpha=0.5$ would be singular (\cref{rem:degenerate}).

\section{Numerical Experiments}
\label{sec:results}

We evaluate the framework in four experiments, organized around four research questions:
\begin{description}
  \item[RQ1] Does PATP-MaxEnt restore solver feasibility and useful body reconstruction for a heavy-tailed law where classical monomial MaxEnt is infeasible? (\cref{sec:exp1})
  \item[RQ2] How do the parity-matched PATP family, T-MaxEnt, and polynomial baselines compare on a multimodal target, and does the one-dimensional $\alpha$-scan select a good member? (\cref{sec:exp2})
  \item[RQ3] Across tail classes (Cauchy, Student-$t$, $\alpha$-stable, Gaussian), what can a logarithmic element recover when the algebraic-tail class is supplied, and what is the cost when that assumption is wrong? (\cref{sec:genelement})
  \item[RQ4] Can the published Pearson and GOPoly rows of \citet{rajan2018moment} be reproduced from the printed protocol, and what happens when GOPoly, T-MaxEnt, and PATP use identical supports and constraint budgets? How should separate heavy-tail and real-data panels be interpreted outside that population-moment benchmark? (\cref{sec:headtohead})
\end{description}

\paragraph{Reproducibility.} All computations use R~4.5.3, single-threaded, on an Apple M-series CPU. Experiments 1--3 use base R; Experiment 4 additionally uses \texttt{PearsonDS} \citep{pearsonDS}, and its real-data panel uses the base-R \texttt{EuStockMarkets} series \citep{rcore}. The supplementary selector check uses \texttt{diptest} \citep{hartigan1985dip}. Experiments 1--3 use seed $20260612$ for the single-realization tables, and Experiments 1--2 additionally use seeds $1$--$20$ for replication (\cref{sec:repli}); the heavy-tail panel uses seeds $20260701$--$20260710$. Experiments 1--3 use the stated rectangle grids and the solver of \cref{alg:newton}. The Rajan audit instead uses analytical moments and quantiles, 480-point Gauss--Legendre quadrature, a 50,001-point inversion grid, and dynamic reorthogonalization. Its generated files retain per-level errors, alternative aggregations, convergence records, and the complete PATP scan. The scripts and tests are provided as supplementary material (see Data and Code Availability).

\subsection{Experiment 1 (RQ1): heavy-tailed law --- standard Cauchy}
\label{sec:exp1}
We generate $N = 1000$ samples from a standard Cauchy distribution and fit $n=4$ constraints, with the fitted density and all integrals supported on the truncated grid $[-50, 50]$ ($2000$ points, $\Delta x \approx 0.05$). The constraint targets are \emph{empirical averages over the full sample} --- every one of the $N=1000$ draws contributes $\varphi_i(x_j;\alpha)$, including any draw with $|x_j|>50$; this is precisely why a target such as $\hat m_4$ can exceed the maximum value attainable on $D$ and so render the monomial problem infeasible (the same convention is used throughout, \cref{sec:patp}). For the Cauchy law the population counterparts of $\varphi_1$ (and, for $\alpha \ge 0.5$, of the higher constraints) do not exist, so the fit is a sample-conditional surrogate. The true Cauchy law places only $1.27\%$ of its mass outside $[-50,50]$. Accuracy is measured against the true law by (i) the body MSE, the mean squared density error on $[-10,10]$; (ii) the Kolmogorov--Smirnov (KS) distance $\sup |F_{\mathrm{fit}} - F_{\mathrm{true}}|$ on $[-45,45]$; and (iii) quantile estimates $\hat q_\beta$ read off the fitted CDF at the nearest grid point, with signed relative errors $(\hat q_\beta - q_\beta)/q_\beta$ against the exact Cauchy quantiles $q_{0.95} = 6.31$, $q_{0.99} = 31.82$.

\begin{table}[t]
  \centering
  \caption{Parity-matched PATP-MaxEnt (PM-PATP) converges on standard Cauchy data for $\alpha \le 0.3$ and reconstructs the central body, whereas the classical monomial member ($\alpha=1$) is infeasible; equal-budget T-MaxEnt also converges but has lower body accuracy. All rows use four constraints, $N=1000$, support $D=[-50,50]$, seed $20260612$. Conv.: solver convergence status; $\kappa_H$: coordinate-specific condition number of the ridge-regularized dual Hessian ($H+10^{-8}I$); body MSE: density MSE on $[-10,10]$; KS: Kolmogorov--Smirnov distance on $[-45,45]$; $\hat q_\beta$ (err): quantile estimate and signed relative error against the exact Cauchy quantile ($q_{0.95}=6.31$, $q_{0.99}=31.82$). The selection criterion \cref{eq:opt_alpha} attains its minimum over the converged PATP rows at $\alpha=0.3$ ($\Gamma=2.395$, bold).}
  \label{tab:cauchy}
  \setlength{\tabcolsep}{1.5pt}
  \resizebox{\linewidth}{!}{%
  \begin{tabular}{llcccccc}
    \toprule
    Basis & param & Conv. & $\kappa_H$ & body MSE & KS & $\hat q_{0.95}$ (err) & $\hat q_{0.99}$ (err) \\
    \midrule
    PM-PATP & $\alpha{=}0.0$ & yes & $7.45 \times 10^{3}$ & $2.46 \times 10^{-3}$ & $0.092$ & $16.8$ ($+166.6\%$) & $41.2$ ($+29.5\%$) \\
    PM-PATP & $\alpha{=}0.1$ & yes & $5.30 \times 10^{3}$ & $2.82 \times 10^{-3}$ & $0.098$ & $17.6$ ($+179.3\%$) & $41.8$ ($+31.3\%$) \\
    PM-PATP & $\alpha{=}0.2$ & yes & $3.30 \times 10^{3}$ & $3.84 \times 10^{-3}$ & $0.119$ & $21.2$ ($+236.3\%$) & $44.6$ ($+40.2\%$) \\
    PM-PATP & $\boldsymbol{\alpha{=}0.3}$ & yes & $2.46 \times 10^{3}$ & $5.58 \times 10^{-3}$ & $0.161$ & $32.6$ ($+416.2\%$) & $47.9$ ($+50.7\%$) \\
    PM-PATP & $\alpha{=}0.4$ & \multicolumn{6}{l}{solver failure (targets attainable on $D$; Newton overflow)} \\
    PM-PATP & $\alpha{=}0.5^{\dagger}$ & yes & $1.71 \times 10^{10}$ & $3.90 \times 10^{-3}$ & $0.261$ & $16.6$ ($+163.4\%$) & $27.6$ ($-13.1\%$) \\
    PM-PATP & $\alpha{=}0.6$ & \multicolumn{6}{l}{solver failure (targets attainable on $D$; Newton overflow)} \\
    PM-PATP & $\alpha{=}0.7$--$0.9$ & \multicolumn{6}{l}{infeasible (high-order targets exceed the attainable range on $D$)} \\
    Monomial & $\alpha{=}1.0$ & \multicolumn{6}{l}{infeasible: $\hat m_4 = 6.9{\times}10^{9} \gg \max_{x\in D} x^4 = 6.25{\times}10^{6}$} \\
    \midrule
    T-MaxEnt & $p{=}0.2, S{=}2$ & yes & $1.13 \times 10^{1}$ & $4.33 \times 10^{-3}$ & $0.316$ & $33.9$ ($+437.6\%$) & $43.6$ ($+37.2\%$) \\
    \bottomrule
  \end{tabular}%
  }

  \smallskip
  \raggedright\footnotesize $^{\dagger}$Degenerate: at $\alpha = 0.5$ all exponents equal $1$ and the design matrix has rank $2$ (\cref{rem:degenerate}); the row is shown for completeness and excluded from comparisons. $\Gamma$ for the T-MaxEnt row uses a different constraint functional and is not comparable to the PATP rows, so it is omitted.
\end{table}

\Cref{tab:cauchy} summarizes the sweep over $\alpha \in \{0, 0.1, \dots, 1.0\}$. Three findings follow.

\emph{Feasibility.} The solver converges for $\alpha \in \{0, 0.1, 0.2, 0.3\}$, where all fractional exponents satisfy $p_i(\alpha) < 1$ and the corresponding population moments of the Cauchy law are finite. At $\alpha = 1$ the monomial problem is \emph{infeasible} on the truncated support: the empirical fourth-moment target $\hat m_4 = 6.9 \times 10^{9}$ vastly exceeds the largest value of $x^4$ attainable on $[-50,50]$, namely $6.25 \times 10^{6}$ (the second moment likewise, $\hat m_2 = 3181 > 2500$), so no density on $D$ can match the target --- this, not generic ``divergence,'' is the precise failure mechanism of classical four-moment MaxEnt here. The same certificate fires for $\alpha \in \{0.7, 0.8, 0.9\}$, whose high-order signed targets (up to $2.1 \times 10^{7}$ at $\alpha = 0.9$) exceed the attainable range on $D$. The failures at $\alpha \in \{0.4, 0.6\}$ are by contrast purely numerical: there the targets are finite and attainable, but the Newton iteration overflows --- so convergence is governed jointly by feasibility on the truncated support and solver robustness, not by moment existence alone. This non-monotone pattern is stable: across $20$ seeds, the four-moment monomial target was infeasible in $19$ and $\alpha=0$ converged in $18$ (\cref{sec:repli}).

\emph{Body accuracy vs.\ tail accuracy.} The most accurate member is $\alpha = 0$: body MSE $2.46 \times 10^{-3}$ and KS distance $0.092$ (over $20$ seeds, $0.068 \pm 0.026$) --- a reconstruction of the central body, quantified by these two statistics, from four sample-surrogate constraints (\cref{fig:cauchy}). Tail accuracy is limited: the $99\%$ quantile is overestimated by $29.5\%$ and the $95\%$ quantile by $167\%$, because an exponential-family fit on a truncated support cannot reproduce the $\sim x^{-2}$ tail decay and redistributes the unmatched tail mass inward (the upper-quantile estimates saturate near the truncation boundary $L=50$). This is a structural limitation of few-constraint MaxEnt on heavy tails rather than a defect of the basis: no member of any basis family tested here achieves accurate Cauchy extreme quantiles on $[-50,50]$.

\emph{Selection criterion.} The dual-potential criterion \cref{eq:opt_alpha} picks $\alpha^{\ast} = 0.3$, the \emph{least} accurate of the four converged fits (body MSE $5.58 \times 10^{-3}$ vs.\ $2.46 \times 10^{-3}$ at $\alpha = 0$). It performs well on the multimodal example (\cref{sec:exp2}) but fails again in the Rajan panel (\cref{sec:headtohead}). The variance-targeted and held-out alternatives in the supplement are useful diagnostics, not validated general selectors.

For comparison, the equal-budget T-MaxEnt row uses four ECF constraints, which exist without moment conditions, and every tested $(p,S)$ configuration converges. Its body accuracy is lower (body MSE $4.33 \times 10^{-3}$, KS $0.316$). The reported $\kappa_H$ values are solver diagnostics in different coordinate systems and do not rank the represented families. The evidence here is limited to observed convergence and reconstruction error: PM-PATP is more accurate on the body, while the tested T-MaxEnt configurations avoid undefined population moments.

\begin{figure}[t]
  \centering
  \includegraphics[width=0.92\linewidth]{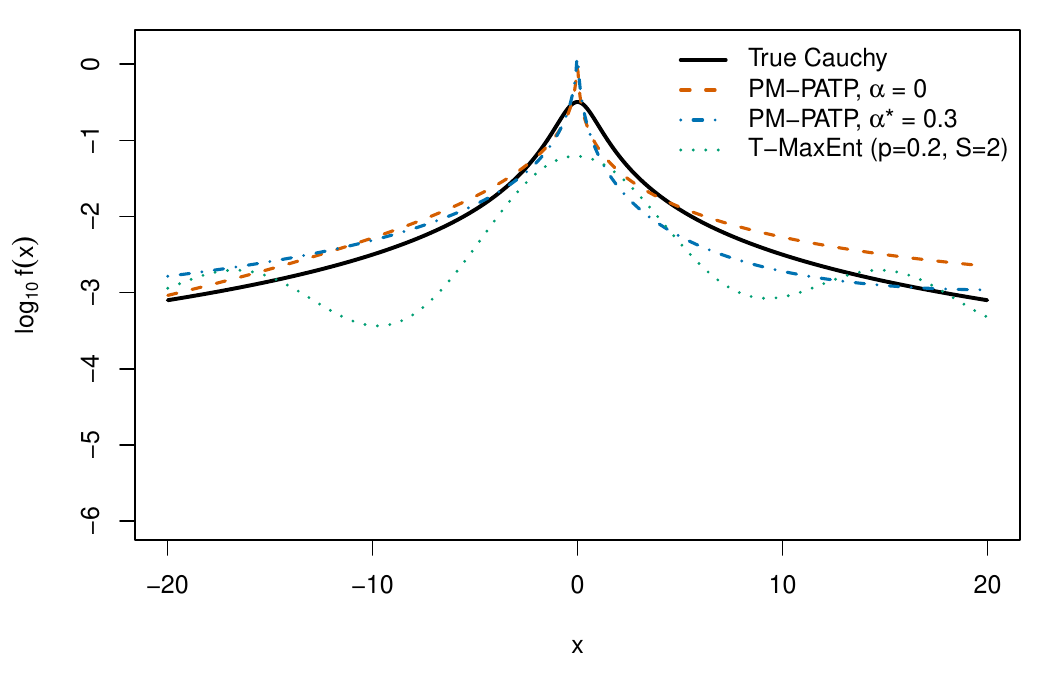}
  \caption{PM-PATP at $\alpha = 0$ tracks the central body of the standard Cauchy density over four decades, while the potential-selected $\alpha^{\ast} = 0.3$ concentrates more mass near the mode and the equal-budget T-MaxEnt ($p=0.2$, $S=2$) flattens the body; all fits underestimate the far tails on the truncated support. This visual ordering matches the body-MSE and KS columns of \cref{tab:cauchy}: $\alpha=0$ is closest to the true density on $[-10,10]$, and all reconstructions diverge from it beyond $|x|\approx 10$. Log-density vs.\ $x$ on $[-20, 20]$; four constraints each, $N = 1000$, seed $20260612$.}
  \label{fig:cauchy}
\end{figure}

\subsection{Experiment 2 (RQ2): multimodal target --- bimodal Gaussian mixture}
\label{sec:exp2}
We fit the bimodal Gaussian mixture
\begin{equation}
  f(x) = 0.4\, \mathcal{N}\!\left(x;\, -1,\, \sigma^2\right) + 0.6\, \mathcal{N}\!\left(x;\, 1,\, \sigma^2\right), \qquad \sigma^2 = 0.32 \ (\sigma \approx 0.566),
  \label{eq:mixture}
\end{equation}
from $N = 1000$ samples (seed $20260612$) using six constraints on the support $[-5, 5]$ (grid of $1000$ points). All constraint targets --- power and trigonometric alike --- are empirical averages over the same sample, so every basis competes on identical information: the same number of constraints, support, grid, solver, and stopping rule. Accuracy is the unweighted MSE between the fitted and true density on the grid (``PDF MSE'').

\begin{table}[t]
  \centering
  \caption{On the bimodal mixture, the PM-PATP member selected by the dual-potential criterion ($\alpha^{\ast} = 0.7$) improves PDF MSE $7.1$-fold over the six-moment monomial baseline on this seed (over $20$ seeds, the ratio is $8.5 \pm 5.8$, favorable in all $20$; \cref{sec:repli}). T-MaxEnt also improves the PDF MSE, while the all-odd Form-B basis illustrates the parity obstruction (\cref{prop:parity}). Six constraints, $N = 1000$, support $[-5,5]$, seed $20260612$. $\kappa_H$ is retained only as a coordinate-specific solver diagnostic; $\Gamma$ is comparable within the PATP family but not across constraint families.}
  \label{tab:mixture}
  \footnotesize\setlength{\tabcolsep}{3.5pt}
  \begin{tabular}{llcccc}
    \toprule
    Basis & Parameters & Conv. & $\kappa_H$ & $\Gamma$ & PDF MSE \\
    \midrule
    PM-PATP (selected) & $\alpha^{\ast} = 0.7$ & yes & $\mathbf{2.43 \times 10^{4}}$ & $\mathbf{1.4198}$ & $8.03 \times 10^{-5}$ \\
    PM-PATP (best MSE) & $\alpha = 0.6$ & yes & $4.00 \times 10^{5}$ & $1.4203$ & $\mathbf{6.32 \times 10^{-5}}$ \\
    PM-PATP & $\alpha = 0.8$ & yes & $7.23 \times 10^{3}$ & $1.4200$ & $9.10 \times 10^{-5}$ \\
    T-MaxEnt & $p = 0.5,\ S = 3$ & yes & $3.24 \times 10^{4}$ & $1.4239$ & $3.48 \times 10^{-4}$ \\
    Monomial ($=$ PM-PATP, $\alpha = 1$) & $x^i,\ i \le 6$ & yes & $1.21 \times 10^{5}$ & $1.4272$ & $5.73 \times 10^{-4}$ \\
    Legendre $P_1$--$P_6(x/5)$ & --- & yes & $1.36 \times 10^{5}$ & $1.4272$ & $5.73 \times 10^{-4}$ \\
    Form-B (all-odd) PATP & $\alpha = 0.9$ & yes & $2.07 \times 10^{10}$ & $2.2400$ & $1.39 \times 10^{-2}$ \\
    \bottomrule
  \end{tabular}

  \smallskip
  \raggedright\footnotesize The Legendre basis spans the same polynomial space as the monomials (the constant is absorbed by normalization), so its fitted density is mathematically identical. Differences in $\kappa_H$ would describe coordinates, not model quality; all listed values are for $H + 10^{-8}I$ in the displayed parameterization.
\end{table}

\Cref{tab:mixture} reports the comparison and \cref{fig:mixture} the fitted densities. The six-moment monomial baseline converges and has PDF MSE $5.73 \times 10^{-4}$. The one-dimensional PATP scan selects $\alpha^{\ast} = 0.7$, with PDF MSE $8.03 \times 10^{-5}$, a $7.1$-fold improvement on identical information and within $27\%$ of the best member in the sweep ($6.32 \times 10^{-5}$ at $\alpha = 0.6$). T-MaxEnt has PDF MSE $3.48 \times 10^{-4}$, a $1.6$-fold improvement over the monomial fit. Condition numbers remain in the table to diagnose these particular solves, but \cref{eq:hessian_congruence} precludes using them as evidence of family-level superiority.

The Form-B row quantifies the parity obstruction of \cref{prop:parity}: the all-odd basis has PDF MSE $1.39 \times 10^{-2}$, $40$ times the T-MaxEnt error and $170$ times the selected PM-PATP error. With only odd constraints, the fitted \emph{log-density} $\ln f = \sum_i \lambda_i \varphi_i$ is odd up to the additive constant $-\ln Z$, so it cannot match a target whose log-density has a non-trivial even part. The bimodal mixture is mildly asymmetric (weights $0.4/0.6$), but its log-density is dominated by an even double-well shape. Comparisons against this structurally restricted row should not be presented as comparisons against general polynomial MaxEnt.

\begin{figure}[t]
  \centering
  \includegraphics[width=0.92\linewidth]{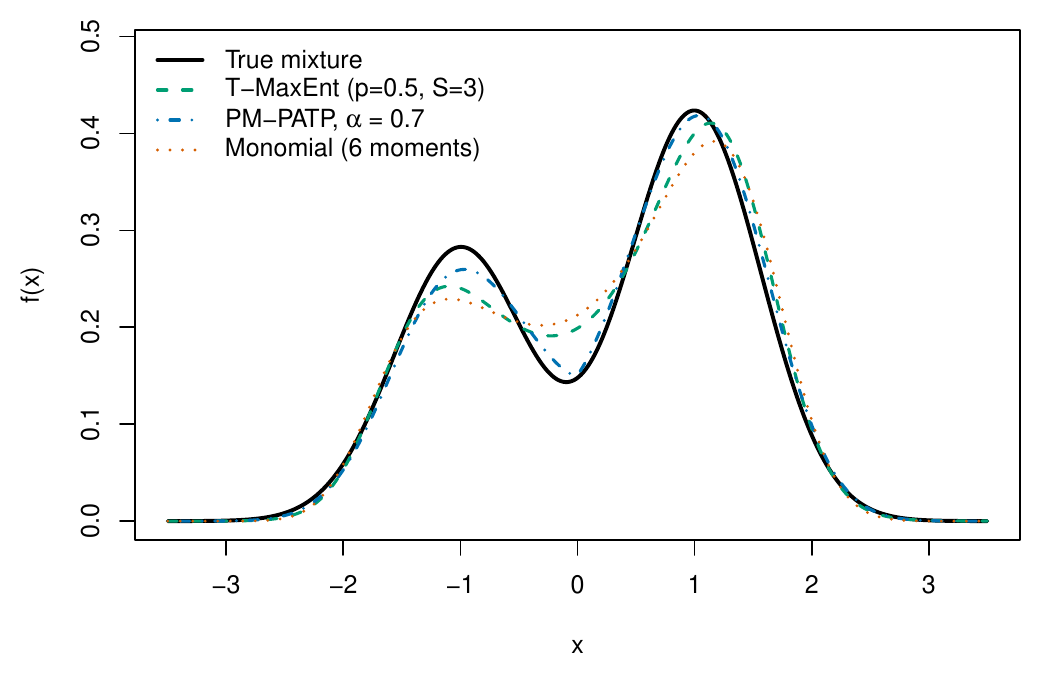}
  \caption{The scan-selected PM-PATP member ($\alpha^{\ast}=0.7$) follows both modes and the central valley of the bimodal mixture; T-MaxEnt ($p = 0.5$, $S = 3$) smooths the modes, and the six-moment monomial fit rounds them further. True density and fits on $[-3.5, 3.5]$; six constraints each, identical sample (seed $20260612$).}
  \label{fig:mixture}
\end{figure}

\paragraph{Frequency ablation and design rule.} The supplementary T-MaxEnt ablation shows that, for fixed $p \le 0.7$, accuracy improves as harmonics are added; the best grid value is PDF MSE $9.67 \times 10^{-5}$ at $(p,S)=(0.7,5)$. The condition numbers recorded in the fixed T-MaxEnt coordinates diagnose increasing numerical redundancy but are not a cross-family criterion. The screen $|\hat\psi_N(Sp)| \ge 0.05$ rejects $(1.0,4)$ and $(1.0,5)$, whose top-harmonic ECF moduli are $0.035$ and $0.009$, but it is not sufficient: the remaining $p=1.0$ configurations are still an order of magnitude less accurate than $p\le0.7$. A $3/\sqrt{N}$ admissibility screen followed by held-out log-score selection chooses $p\in\{0.5,0.7\}$ in all ten seeds, with mean PDF-MSE penalty $1.7\times$ the per-seed oracle versus $3.3\times$ for the fixed $(0.5,3)$ default. This is a within-family tuning result, not a validated selector among generating-element families.

\paragraph{Does the PATP map matter, or any one-parameter family?} Tying the exponents to one scalar replaces a multidimensional FM-MEM search by a fixed scan. A free five-exponent Nelder--Mead search requires 52--154 inner solves, depending on its start, and gives PDF MSE $9.5$--$10.5\times10^{-5}$; the 11-point PATP scan gives $8.0\times10^{-5}$. Over five seeds, a genetic-algorithm search uses about 209 inner solves, compared with 10 for the scan, and the same dual-potential selection gives mean PDF MSE $1.0\times10^{-4}$ versus $9.7\times10^{-5}$. When both searches use the unavailable truth MSE, the genetic algorithm gives $5.9\times10^{-5}$ and the PATP grid $7.6\times10^{-5}$. On this example, the additional search cost changes little under a realizable selector, although the oracle result shows some remaining family capacity. A linear path $q_i(\alpha)=(1-\alpha)/i+\alpha i$ gives selected PDF MSE $8.5\times10^{-5}$, indicating that the main computational reduction comes from the one-parameter structure rather than uniquely from the quadratic PATP map.

\subsection{Experiment 3 (RQ3): the matching principle across tail classes}
\label{sec:genelement}
We now vary the generating element against four targets of differing tail class: standard Cauchy (algebraic tail, index $-2$), Student-$t$ with two degrees of freedom ($-3$), symmetric $\alpha$-stable with $\alpha=1.5$ ($-2.5$; samples by the Chambers--Mallows--Stuck method, reference density by trapezoidal characteristic-function inversion), and standard Gaussian (light tail, control). All use $N=1000$ (seed $20260612$), the support $[-L,L]$ with $L=50$ (heavy) or $L=8$ (Gaussian), and the same solver. We compare five elements: the agnostic fractional element (power-PATP, $\alpha=0$, four constraints), the trigonometric element (T-MaxEnt, $p=0.2$, $S=2$), the matched logarithmic element with one scale (LogRat-1, one constraint) and four scales (LogRat-multi, $s\in\{0.5,1,2,4\}$), and a bounded-rational control ($1/(1+(x/s)^2)$, same scales). Fit quality is the KS distance on the support interior; tail fidelity is the slope of $\log_{10} f$ versus $\log_{10}|x|$ on the tail (the body-weighted KS does \emph{not} reveal tail behavior, so both are reported).

\begin{table}[t]
  \centering
  \caption{Tail-matched generating elements recover algebraic tails that agnostic and bounded controls distort. KS distance (lower better) by element and target, plus the recovered Cauchy tail slope (true $-2$). On algebraic-tailed targets the matched logarithmic element is best or near-best in KS \emph{and} is the only element that recovers the tail index; the bounded-rational control attains comparable KS but a spuriously flat tail (slope $-0.62$), so KS alone does not certify tail fidelity; on the light-tailed Gaussian the trigonometric and agnostic elements do well while the single-scale logarithmic element does not. Bold marks the recommended/matched element per regime, not the strict column minimum. KS is computed on $[-0.9L,0.9L]$ over a $4000$-point grid; small differences from \cref{tab:cauchy} (e.g.\ power-PATP Cauchy $0.087$ here vs.\ $0.092$ there, $2000$-point grid) are grid-resolution effects. $N=1000$, seed $20260612$.}
  \label{tab:genelement}
  \footnotesize
  \begin{tabular}{lccccc}
    \toprule
    Generating element & \multicolumn{4}{c}{KS distance} & Cauchy \\
    \cmidrule(lr){2-5}
    & Cauchy & $t(2)$ & S$\alpha$S $1.5$ & Gaussian & tail slope \\
    \midrule
    power-PATP $\alpha{=}0$ (agnostic) & $0.087$ & $0.040$ & $0.036$ & $0.042$ & $-1.14$ \\
    T-MaxEnt (trigonometric)          & $0.316$ & $0.332$ & $0.331$ & $\mathbf{0.013}$ & $+1.85$ \\
    LogRat-1 (matched, $1$ constr.)   & $0.015$ & $0.041$ & $0.039$ & $0.052$ & $\mathbf{-1.88}$ \\
    LogRat-multi (matched, $4$)       & $0.012$ & $\mathbf{0.028}$ & $\mathbf{0.0035}$ & $0.012$ & $-0.99^{\ddagger}$ \\
    BoundedRat (control, $4$)         & $\mathbf{0.012}$ & $0.027$ & $0.0052$ & $0.012$ & $-0.62^{\ddagger}$ \\
    \bottomrule
  \end{tabular}

  \smallskip
  \raggedright\footnotesize S$\alpha$S$(1.5)$: symmetric $\alpha$-stable, $\alpha=1.5$. True tail slopes: Cauchy $-2$, $t(2)$ $-3$, S$\alpha$S$(1.5)$ $-2.5$. $^{\ddagger}$Multi-constraint elements give an unreliable single tail-slope estimate (scale mixing); the clean estimate is LogRat-1's, whose fitted density has analytic slope $2\hat\lambda$ as $|x|\to\infty$ --- here $\hat\lambda=-0.943$ on Cauchy (cf.\ $-1$ for the true law), giving $-1.886$; the tabulated $-1.88$ is the slope measured on the finite window $[8,40]$. The single-constraint $\kappa_H=1$ is a $1{\times}1$ Hessian and is \emph{not} a conditioning claim; the multi-scale variants have $\kappa_H\sim10^4$, comparable to the other elements. The Gaussian-core hybrid $\{x^2,\log(1+x^2)\}$ is omitted: it is infeasible on Cauchy, inheriting the divergent $\hat m_2$ target that breaks monomial MaxEnt on infinite-variance laws.
\end{table}

\Cref{tab:genelement} and \cref{fig:genelement} support the matching principle, subject to three qualifications.

\emph{What matching can recover when the class is supplied.} On the three algebraic-tailed targets the truth-matched logarithmic element is best or near-best in KS (LogRat-multi: $0.012$, $0.028$, and $0.0035$ on Cauchy, $t(2)$, and stable, respectively --- essentially tied with the bounded-rational control). Its distinct representational property is the tail index: $\exp(\lambda\log(1+x^2))=(1+x^2)^\lambda$ has asymptotic tail slope exactly $2\lambda$, and the single-constraint fit returns $\hat\lambda=-0.943$ on standard Cauchy (measured slope $-1.88$ on $[8,40]$, within $\approx6\%$ of the true $-2$). The fractional and trigonometric fits return slopes $-1.14$ and $+1.85$, respectively. These are oracle-family comparisons because the logarithmic tail class and scale are supplied, not selected from the data.

\emph{KS hides tail failure; the bounded-rational control proves it.} The bounded-rational element attains KS as low as the logarithmic one (Cauchy $0.012$), yet its tail slope is a spuriously flat $-0.62$: $\exp(\lambda/(1+x^2))\to\mathrm{const}$ as $|x|\to\infty$, so it flattens rather than reproduces the tail. Body-weighted KS does not see this; the tail-slope diagnostic does. Hence a low KS is necessary but not sufficient evidence of a correct heavy-tail reconstruction.

\emph{A committed matched element is the wrong tool off its class.} On the light-tailed Gaussian the \emph{single-scale} logarithmic element (LogRat-1) is the worst performer (KS $0.052$, forcing an algebraic tail onto a Gaussian), while the trigonometric element excels (KS $0.013$) and the agnostic fractional element is adequate ($0.042$). The multi-scale logarithmic element is more flexible and stays competitive even here (KS $0.012$), but only because its several scales let it mimic a light core --- at which point it has forfeited the clean single-parameter tail-index interpretation that is its distinctive value. This is the matching principle's necessary converse: a \emph{committed} matched element is a two-sided assumption. In this four-target, single-seed panel, the agnostic fractional element has KS values from $0.036$ to $0.087$ and does not have the largest KS value for any target. This pattern is descriptive, not a guarantee of the construction or evidence for a universally preferred default when the tail class is unknown.

\emph{The matched element escapes truncation in closed form.} The quantile saturation that limits the truncated-grid fits (\cref{sec:limitations}) is a property of the bounded support, not of the matched element. On the whole line $\mathbb{R}$ the model $(1+(x/s)^2)^\lambda$ is the Student/Cauchy family with closed-form normalization $Z=s\,B(\tfrac12,-\lambda-\tfrac12)$ and closed-form moment $\psi(-\lambda)-\psi(-\lambda-\tfrac12)$ (\cref{sec:lograt}), so the single-constraint fit is an exact one-dimensional root-find and quantiles are read from the analytic Student-$t$ CDF. On standard Cauchy this returns $\hat\lambda=-0.989$ (analytic slope $2\hat\lambda=-1.98$, versus $-1.89$ from the $[-50,50]$ grid fit) and a KS over all of $\mathbb{R}$ of $0.0035$ (versus $0.0147$ truncated). It also removes quantile saturation: the truncated grid underestimates $q_{0.99}$ at $23.5$ (true $31.8$, $-26\%$), whereas the unbounded fit returns $34.5$ ($+8.5\%$). The residual error is the single-constraint, fixed-scale limitation --- it does not transfer cleanly to $t(2)$/$t(3)$, whose scale differs from $s=1$ --- rather than a truncation artifact.

\begin{figure}[t]
  \centering
  \includegraphics[width=0.92\linewidth]{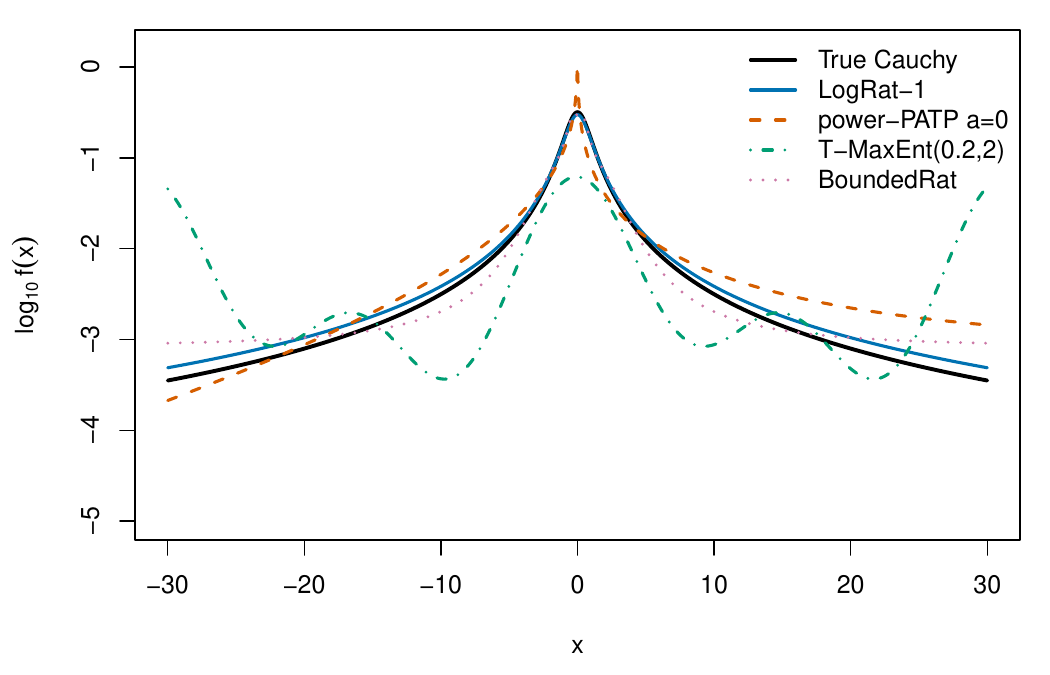}
  \caption{When the Cauchy tail class is supplied, LogRat-1 tracks the $x^{-2}$ tail over four decades. The fractional fit decays too slowly, the truncated trigonometric fit has no algebraic-tail form, and the bounded-rational control flattens. Log-density vs.\ $x$ on $[-30,30]$, $N=1000$, seed $20260612$; this is an oracle-family illustration, not a data-driven selection result.}
  \label{fig:genelement}
\end{figure}

\subsection{Experiment 4 (RQ4): Rajan benchmark audit and domain extension}
\label{sec:headtohead}
We first reconstruct the six Appendix-B distributions of \citet{rajan2018moment} from their analytical population moments. In the labelling used below, D1 is the three-component beta mixture of their Eq.~(15); D2--D5 are two-component normal mixtures (their Eqs.~(16)--(19)) with component means $\pm4$, $\pm1$, $\{0,3\}$ and $\{0,25\}$ and standard deviations $\{4,3\}$, $0.4$, $1$ and $1$; and D6 is the symmetric mixture of two shifted beta densities on $[-1,1]$ of their Eq.~(20). For each distribution and budget $n\in\{4,8,12\}$, the support is obtained from the two roots of the Christoffel mass in Eq.~(13) of that paper, with $\tau=10^{-(\lfloor n/2\rfloor+1)}$. Truth quantiles are analytical, and the evaluation uses all 30 printed levels, from $0.001\%$ through $99.999\%$. At level $p_j$ we compute
\begin{equation}
  e_j=100\,\frac{|\hat q(p_j)-q(p_j)|}{|q(p_j)-m_1|}.
  \label{eq:rajan_error}
\end{equation}
For each distribution, \citet{rajan2016benchmark} define $\varepsilon=\max(\bar e_{\mathrm{lower}},\bar e_{\mathrm{upper}})$. The table mean is the mean of the six $\varepsilon$ values; its reported SD is reconstructed as the SD of all $6\times30$ pointwise errors. We retain alternative aggregations in the supplement because the wording in the 2018 paper is not explicit.

GOPoly uses the monomial constraint span, the printed support rule, Gauss--Legendre quadrature, and Gram--Schmidt orthogonalization both initially and at each Newton iteration. T-MaxEnt and PATP use the same support, quadrature, and solver. The T-MaxEnt arm fixes the fundamental frequency at $2\pi$ divided by the support width. PATP scans $\alpha=0,0.1,\ldots,1$; the rank-deficient $\alpha=0.5$ member is excluded rather than regularized. Away from this midpoint, the PATP columns are evaluated as parity-block divided differences. This invertible triangular coordinate change preserves the constraint span, while avoiding cancellation among fractional-power columns; convergence residuals are also replayed in the original PATP coordinates. This is an independent implementation of the printed algorithm. The original MATLAB download is no longer available, so an unresolved difference from the authors' code cannot be ruled out.

\paragraph{Reproduction audit.} The reconstructed aggregation reproduces the Pearson row: $11.487\%$ mean and $10.760\%$ SD, compared with $11.46\%$ and $10.77\%$ in Table~2. It does not reproduce any GOPoly row (\cref{tab:rajan-reproduction}). The discrepancy is concentrated in D1, the three-component beta mixture: its consolidated GOPoly error is $70.743\%$, $22.229\%$, and $10.515\%$ for $n=4,8,12$. At $n=12$, the other five distributions range from $0.0036\%$ to $0.681\%$. Increasing the quadrature order from 240 to 960 leaves the D1 $n=4$ result unchanged at the shown precision. Replacing the estimated support by the analytical support changes the result but neither recovers the published row nor follows the printed support policy. The proper conclusion is therefore non-reproduction under the literal protocol, not a corrected value for the original study.

\begin{table}[t]
  \centering
  \caption{Audit of Table~2 in \citet{rajan2018moment}. Entries are mean/SD (\%). Pearson is reproduced; the three GOPoly rows are not reproduced by the literal printed protocol.}
  \label{tab:rajan-reproduction}
  \small
  \begin{tabular}{lccc}
    \toprule
    Method & $n$ & Published & This audit \\
    \midrule
    Pearson & 4  & $11.46/10.77$ & $11.487/10.760$ \\
    GOPoly  & 4  & $3.55/4.20$   & $14.771/23.811$ \\
    GOPoly  & 8  & $1.00/1.73$   & $4.706/8.168$ \\
    GOPoly  & 12 & $0.34/0.58$   & $2.031/4.178$ \\
    \bottomrule
  \end{tabular}
\end{table}

\paragraph{Equal-support comparison.} \Cref{tab:rajan-equal-support} compares the three generating-element families under the same operational support and constraint budget. All entries contain six converged targets. PATP-oracle is the truth-informed lower envelope over the $\alpha$ grid and is included only to measure family capacity. The selected arm minimizes the converged dual potential. The comparison is mixed: the oracle envelope is lower at $n=8$ and $n=12$, whereas the selected arm retains substantial distribution-specific gaps. At $n=12$, selected/oracle errors are $0.885/0.0288\%$ on D2, $0.905/0.102\%$ on D4, and $0.130/0.0391\%$ on D6. On D5 the rule selects $\alpha=0.7$ with error $0.000428\%$, close to the $\alpha=0.4$ oracle value $0.000350\%$. No deployable ranking follows from these data.

\begin{table}[t]
  \centering
  \caption{Mean of the six consolidated errors $\varepsilon$ (\%) under identical Rajan supports and budgets. PATP-selected uses the dual-potential heuristic; PATP-oracle uses true quantile error and is not operational. Every cell has 6/6 converged targets.}
  \label{tab:rajan-equal-support}
  \small
  \begin{tabular}{rcccc}
    \toprule
    $n$ & GOPoly & T-MaxEnt & PATP-selected & PATP-oracle \\
    \midrule
    4  & $14.771$ & $15.108$ & $11.518$ & $10.923$ \\
    8  & $4.706$  & $5.115$  & $2.683$  & $2.621$ \\
    12 & $2.031$  & $2.370$  & $1.421$  & $1.129$ \\
    \bottomrule
  \end{tabular}
\end{table}

\paragraph{Continuous-density replay.} A raw-power ablation for D5 at $n=12$ and $\alpha=0$ passed the discrete Newton residual but produced an unresolved spike at zero and a spurious $99.996\%$ error. Direct divided-difference evaluation gave $0.001529\%$ for the same PATP span. We checked this trigger independently by adaptive integration, rather than by the solver quadrature or the uniform quantile grid: the maximum original-coordinate constraint residual was $2.69\times10^{-9}$, the adaptive error was $0.001534\%$, and the largest adaptive-versus-grid quantile difference was $4.20\times10^{-7}$. The raw ablation is therefore treated as numerical false convergence, not as selector evidence.

\paragraph{Heavy-tail domain extension.} The synthetic heavy-tail panel is separate from the Rajan reproduction. Cauchy, Student-$t(2)$, and S$\alpha$S$(1.5)$ have no finite variance; Student-$t(3)$ has no fourth moment. Consequently, the population constraints required by four-moment Pearson and GOPoly are undefined. Their entries in \cref{tab:hh-heavy} are finite-sample surrogates built from empirical moments, not in-domain evaluations of those methods. LogRat is deliberately supplied with the algebraic-tail family and four scales, so its lower errors are matched-family oracle evidence. They show that population-defined logarithmic constraints extend the usable domain; they do not show that an automatic procedure would choose LogRat.

\paragraph{Real measured data.} The bottom row of \cref{tab:hh-heavy} is a descriptive comparison on the daily log-returns of DAX, SMI, CAC, and FTSE in \texttt{EuStockMarkets} ($n=1859$ each). Mean tail-quantile errors are $3.1\%$ for the manually matched LogRat arm, $3.0\%$ for Pearson, $7.2\%$ for the raw monomial surrogate, and $9.4\%$ for PATP. Without resampling intervals or a prevalidated family selector, the $3.1\%$ versus $3.0\%$ difference has no ranking interpretation.

\begin{table}[t]
  \centering
  \caption{Tail-quantile error (\%) at the $90/95/99\%$ levels, mean over ten seeds (synthetic) or over the evaluated levels (real data). $^\dagger$LogRat is matched in advance to an algebraic-tail family. $^\ast$Pearson and raw monomial entries on the synthetic panel are empirical-moment surrogates outside their population-moment domain; the monomial column is not labeled GOPoly.}
  \label{tab:hh-heavy}
  \footnotesize\setlength{\tabcolsep}{3pt}
  \begin{tabular}{lcccc}
    \toprule
    Target & LogRat$^\dagger$ & PATP & Pearson$^\ast$ & raw mono.$^\ast$ \\
    \midrule
    \multicolumn{5}{l}{\emph{Heavy-tailed synthetic targets ($N=1000$, ten seeds)}}\\
    Cauchy                 & $16.7$ & $28.4$ & $161.7$ & $60.1$ \\
    Student-$t(3)$         & $3.4$  & $7.4$  & $15.0$  & $19.1$ \\
    Student-$t(2)$         & $5.7$  & $11.4$ & $56.8$  & $36.4$ \\
    S$\alpha$S $(\alpha{=}1.5)$ & $7.1$ & $12.7$ & $74.9$ & $37.4$ \\
    \midrule
    \multicolumn{5}{l}{\emph{Real data (mean of four index series)}}\\
    EuStockMarkets returns & $3.1$ & $9.4$ & $3.0$ & $7.2$ \\
    \bottomrule
  \end{tabular}
\end{table}

\paragraph{Selection boundary.} Neither selector tested in the supplementary study closes the oracle-to-deployment gap. The entropy rule has the D2, D4, and D6 gaps reported above. The held-out family selector routes the close bimodal target to LogRat in all eight replications and splits its decisions on Cauchy. Matched LogRat and PATP-oracle are therefore reported as family-capacity results. A reliable data-driven choice among generating-element families remains open.

\subsection{Replication over seeds (Experiments 1--2)}
\label{sec:repli}
Because Experiments 1 and 2 each use a single data realization, we repeated both over $20$ independent seeds ($s = 1, \dots, 20$). On the Cauchy benchmark the comparison with the monomial member depends on how its moment targets are formed, and we report both conventions. If the targets use the \emph{full} sample while the density is fitted on $[-50,50]$, the four-moment monomial target is unattainable ($\hat m_4 > \max_{x \in D} x^4$) in $19$ seeds. If targets use only retained draws, the certificate cannot fire because $x^4 \le L^4$ for each retained draw; this discards $1.2\%$ of observations on average, and the monomial fit converges in all $20$ seeds. Under this more favorable convention, monomial KS is $0.194 \pm 0.011$, compared with $0.036 \pm 0.009$ for PATP at $\alpha=0$. Under the full-sample convention, PATP at $\alpha=0$ converges in $18$ seeds, with KS $0.068 \pm 0.026$ and body MSE $(2.04 \pm 0.87) \times 10^{-3}$. On the mixture benchmark, the dual-potential criterion selects $\alpha^{\ast}\in\{0.6,0.7,0.8\}$ in $19$ seeds (median $0.7$); its PDF-MSE improvement over the six-moment monomial baseline is $8.5 \pm 5.8$-fold (range $1.1$--$21.7$), favorable in all $20$. These replications support the stated body-error comparisons, but not a general selection claim: substantial selected-to-oracle gaps remain elsewhere (\cref{sec:headtohead}). Experiments 1--2 have 20-seed summaries (mean $\pm$ SD), the supplementary tuning studies use 10--20 seeds, the FM-MEM control uses five, and Experiment 3 is a single-seed analytic illustration (\cref{sec:limitations}).

\section{Discussion and Design Guidelines}
\label{sec:discussion}

The experiments distill into a \textbf{design map} for the generating element, summarized in \cref{tab:designmap}: the element must be matched to what is known about the target. The map rests on \emph{analytic} mechanisms --- the parity obstruction (\cref{prop:parity}), the exact tail-slope identity ($\text{slope}=2\lambda$), and the bounded-rational flattening --- illustrated by the single-seed numerical examples of \cref{tab:genelement}. The 20-seed results of \cref{sec:repli} cover only the Cauchy and mixture studies. They are consistent with the corresponding entries of the map but do not independently validate its other regimes. We therefore present the map as a set of analytic demonstrations with illustrative numerics, not as a multi-seed empirical ranking. Each entry below states a claim and its supported scope.

\begin{table}[t]
  \centering
  \caption{Design map: which generating element for which target. ``Agnostic'' elements assume no tail class; the ``matched'' logarithmic element assumes an algebraic tail and pays for the assumption when it is wrong.}
  \label{tab:designmap}
  \small
  \begin{tabular}{>{\raggedright\arraybackslash}p{0.26\linewidth}>{\raggedright\arraybackslash}p{0.26\linewidth}>{\raggedright\arraybackslash}p{0.38\linewidth}}
    \toprule
    Target regime & Candidate element & Supported rationale \\
    \midrule
    Unknown / heavy body & fractional-power (PATP), agnostic & some sub-unit moments exist where higher integer moments do not; contains monomials at $\alpha{=}1$ \\
    Bounded, multimodal, light tail & trigonometric (CF) & bounded constraints exist for any law; frequency still requires tuning \\
    Known algebraic (power-law) tail & logarithmic--rational, matched & directly represents and identifies a power-law exponent \\
    Light-tailed, unimodal, compact & classical monomial (low order) & adequate baseline in the compact cases studied here \\
    \bottomrule
  \end{tabular}
\end{table}

\begin{enumerate}
  \item \textbf{Feasibility on heavy tails (fractional element).} On the Cauchy benchmark, fractional parity-matched constraints with $p_i(\alpha) < 1$ restore solver feasibility where the four-moment monomial problem is infeasible on the truncated support (\cref{tab:cauchy}). The supported claim is feasibility plus body reconstruction (KS $0.092$); extreme-quantile accuracy is \emph{not} achieved by any agnostic element at this constraint budget and support, and we do not claim it.
  \item \textbf{Multimodal reconstruction (PM-PATP and T-MaxEnt).} Against a monomial baseline on identical information, selected PM-PATP improves PDF MSE $7.1$-fold on the displayed seed and $8.5 \pm 5.8$-fold over $20$ seeds; T-MaxEnt improves it $1.6$-fold on the displayed seed. The $\kappa_H$ values attached to these fits are coordinate-specific numerical diagnostics, not a second performance axis.
  \item \textbf{Choosing between the families.} PATP contains the monomials at $\alpha=1$, so its truth-oracle envelope cannot be worse than that member. This statement concerns family capacity only. The dual-potential rule is suboptimal on Cauchy and retains selected-to-oracle gaps on Rajan D2, D4, and D6 at $n=12$; T-MaxEnt also requires within-family frequency selection. The experiments do not validate an automatic choice between PATP, T-MaxEnt, and LogRat.
  \item \textbf{The matched logarithmic element, and the cost of matching (\cref{sec:genelement}).} If an algebraic tail and scale are supplied, $\log(1+(x/s)^2)$ gives $(1+(x/s)^2)^\lambda$ and identifies tail slope $2\lambda$. On standard Cauchy, the single-constraint fit gives a measured tail slope $-1.88$ against the true $-2$. This is an oracle-family result. On Gaussian data the same single-scale choice is the worst tested arm, and a bounded-rational control attains a similar KS while producing a false tail slope. Matching is therefore an assumption whose value must be evaluated with tail-sensitive diagnostics.
  \item \textbf{Frequency selection for T-MaxEnt.} The highest harmonic should remain above the ECF noise floor, but this is only a screen: $p=1.0$ can pass while producing an error an order of magnitude above $p\le0.7$. Adding harmonics can also increase numerical redundancy in the stated coordinates. Held-out tuning is useful within T-MaxEnt but does not select among element families.
  \item \textbf{The $\alpha$-selection heuristic.} Minimizing the converged dual potential selects a near-optimal member on the bimodal example, but it is suboptimal on Cauchy and differs materially from the oracle on several Rajan targets. The variance-targeted and held-out alternatives each answer narrower questions and also fail as general family selectors. Until a prespecified rule is validated across regimes, report the full scan and distinguish selected from oracle results.
\end{enumerate}

\section{Limitations}
\label{sec:limitations}

\textbf{Scope.} All methods and experiments are univariate; multivariate extension is future work.

\textbf{Truncated-support dependence.} For signed bases the partition function is finite only on a bounded support (\cref{sec:dual}), so the truncation length $L$ is part of the model; all Cauchy-experiment conclusions for those bases are conditional on $L = 50$, and quantile estimates are additionally limited by grid resolution. The matched logarithmic element is normalizable on $\mathbb{R}$ and admits a closed-form, truncation-free fit (\cref{sec:lograt,sec:genelement}). The supplementary sensitivity study shows that body metrics and the selected $\alpha^{\ast}$ change less than extreme quantiles as $L$ varies; it does not remove support dependence from the other families.

\textbf{Tail accuracy.} Few-constraint MaxEnt on a truncated support does not recover power-law tails: the Cauchy $95\%$ and $99\%$ quantiles are overestimated by between $+29.5\%$ and $+416\%$ across the converged non-degenerate members (the $\hat q_{0.95}$, $\hat q_{0.99}$ columns of \cref{tab:cauchy}; the degenerate $\alpha=0.5$ row, excluded from comparisons, is the sole exception). Applications requiring accurate extreme quantiles of heavy-tailed measurands need tail-dedicated methods, not this framework as-is.

\textbf{Sample-surrogate constraints.} For the Cauchy law some population constraints (notably $\mathbb{E}[X]$) do not exist; the fitted densities are sample-conditional surrogates. We quantify their variability across $20$ seeds in \cref{sec:repli} (e.g.\ Cauchy body KS $0.068 \pm 0.026$), but a full sampling-distribution characterization (e.g.\ bootstrap intervals on every reported statistic) is left to future work.

\textbf{Selection criterion.} The dual-potential rule \cref{eq:opt_alpha} has no optimality theory. It is suboptimal on Cauchy and leaves selected-to-oracle gaps on Rajan D2, D4, and D6. The supplementary oPMM$\alpha$ criterion is tied to a prespecified target functional, while held-out log score targets predictive shape; neither provides a reliable choice among element families. In the selector check, the close bimodal target is routed to LogRat in all eight replications and Cauchy decisions split. Automatic family selection is unresolved.

\textbf{Single-seed generating-element comparison.} Experiment 3 (\cref{sec:genelement}) uses one realization per target, unlike the $20$-seed Experiments 1--2. Its conclusions are nonetheless structural rather than sample-driven: the tail-index identity (slope $=2\lambda$), the bounded-rational flat-tail mechanism, and the parity obstruction are analytic, and the KS-versus-tail-slope discrepancy is a property of the metrics, not of the draw. A full multi-seed replication of the design map is left to future work.

\textbf{Solver.} \Cref{alg:newton} carries no global convergence guarantee. The D5 raw-power ablation also shows that a small discrete gradient is not a continuous-density certificate when the coordinate map creates a feature narrower than the quadrature spacing. Direct divided differences, replay in the original PATP coordinates, and the independent adaptive check resolve the observed trigger; they do not constitute a global guarantee for every basis and target.

\textbf{Benchmarks.} The Rajan audit implements the printed Christoffel support estimator and dynamic Gram--Schmidt GOPoly solver, rather than using a raw monomial proxy. It reproduces Pearson but not the published GOPoly rows. Because the original MATLAB download is unavailable, we cannot determine whether the remaining D1 discrepancy reflects an undocumented implementation choice, an aggregation difference, or an error in the printed material. The heavy-tail panel is outside GOPoly's population-moment domain and therefore labels its polynomial comparator as a raw empirical-moment surrogate. The FM-MEM controls isolate exponent selection: on the bimodal problem a free Nelder--Mead search uses $52$--$154$ inner solves and a genetic-algorithm search about $209$, against $10$ for the scan, but we do not reimplement their dimensional-reduction moment pipelines \citep{zhang2013entropy,xu2019novel}. These boundaries preclude a general state-of-the-art superiority claim.

\textbf{Generating-element catalog.} The three elements studied are instances, not an exhaustive set: a target with an exponential or sub-exponential (e.g.\ Laplace, Weibull) tail would call for yet another element, and the matched logarithmic element assumes a fixed scale $s$ that must itself be chosen (we stack a few scales, which works but muddies the single-constraint tail-index estimate). A principled, data-driven element-selection rule --- rather than the tail-class reasoning we apply by hand --- is open.

\section{Conclusion}
\label{sec:conclusion}

We formulate moment-constrained MaxEnt as a choice of generating element, which determines the sufficient statistics, population existence requirements, and representable log-density. Hessian condition numbers do not share this invariance: they depend on parameterization and are retained only as solver diagnostics. The parity theorem rules out odd-only elements for non-uniform symmetric targets. PATP supplies a reproducible one-parameter path containing the monomials, trigonometric constraints remain population-defined without moment assumptions, and a supplied logarithmic element represents algebraic tails with slope $2\lambda$.

The numerical evidence is deliberately narrower than the original framing. PATP improves body reconstruction on the Cauchy and bimodal examples, but its entropy selector is unreliable. The Rajan audit reproduces Pearson and fails to reproduce the published GOPoly rows; under identical supports the GOPoly, T-MaxEnt, and selected-PATP results are mixed, while PATP-oracle is non-operational. The heavy-tail panel demonstrates a domain extension because logarithmic constraints remain population-defined when the required integer moments do not. Its LogRat results are matched-family oracle evidence, not evidence for automatic deployment.

Future work has three immediate gates: resolve the D1 GOPoly discrepancy if the original implementation becomes available; validate a prespecified family selector on distributions where the choice is genuinely ambiguous; and extend the formulation to multivariate constraints. Broader comparisons with the 124-distribution suite of \citet{rajan2016benchmark} and complete FM-MEM pipelines are warranted only after those gates are met.

\section*{Ethical Considerations}
This is a methodological contribution with no human-subjects, personal, or sensitive data. The data are synthetic or standard simulated benchmarks, together with the public aggregate \texttt{EuStockMarkets} series distributed with base R. The main misuse risk is over-trust in reconstructed tails: few-constraint MaxEnt on truncated supports does not reproduce heavy-tail extreme quantiles without a tail-matched model and separate validation.

\section*{Data and Code Availability}
The R scripts that generate the tables and figures are available at \url{https://github.com/SZabolotnii/Ku-MaxEnt-code-supplement}. The Rajan audit includes analytical D1--D6 truth functions, all 30 probability levels, the printed support estimator, dynamic GOPoly reorthogonalization, equal-support T-MaxEnt and PATP arms, stable divided-difference PATP coordinates, original-coordinate residual replay, per-level errors, alternative aggregations, and numerical tests. The repository also contains the independent adaptive validation of the D5 numerical trigger and the replication, heavy-tail, real-data, selector, and exponent-path scripts. Core experiments use base R~4.5.3; the Pearson and selector panels use the packages stated in \cref{sec:results}.

\section*{CRediT Authorship Contribution Statement}
\textbf{Serhii Zabolotnii:} Conceptualization, Methodology, Software, Formal analysis, Investigation, Writing --- original draft, Writing --- review \& editing.

\section*{Declaration of Competing Interests}
The author declares no competing financial or personal interests.

\section*{Funding}
This research received no specific grant from any funding agency in the public, commercial, or not-for-profit sectors.

\section*{Declaration of Generative AI and AI-Assisted Technologies in the Manuscript Preparation Process}
During preparation of this work, the author used Anthropic Claude (including Claude Code) and OpenAI Codex to improve language and readability and to assist in developing, checking, and documenting the reproducibility code. The author reviewed and edited the resulting material and takes full responsibility for the content.

\bibliographystyle{plainnat}
\bibliography{references}

@article{hartigan1985dip,
  author  = {Hartigan, J. A. and Hartigan, P. M.},
  title   = {The Dip Test of Unimodality},
  journal = {The Annals of Statistics},
  volume  = {13},
  number  = {1},
  pages   = {70--84},
  year    = {1985},
  doi     = {10.1214/aos/1176346577}
}

@Manual{pearsonDS,
  author  = {Becker, Martin and Kl{\"o}{\ss}ner, Stefan},
  title   = {{PearsonDS}: Pearson Distribution System},
  note    = {R package},
  url     = {https://CRAN.R-project.org/package=PearsonDS},
  year    = {2025}
}

@Manual{rcore,
  author       = {{R Core Team}},
  title        = {R: A Language and Environment for Statistical Computing},
  organization = {R Foundation for Statistical Computing},
  address      = {Vienna, Austria},
  year         = {2026},
  note         = {The \texttt{EuStockMarkets} series is distributed with the base \texttt{datasets} package},
  url          = {https://www.R-project.org/}
}

@article{jaynes1957information,
  author  = {Jaynes, Edwin T.},
  title   = {Information theory and statistical mechanics},
  journal = {Physical Review},
  volume  = {106},
  number  = {4},
  pages   = {620--630},
  year    = {1957},
  doi     = {10.1103/PhysRev.106.620}
}

@article{mead1984maximum,
  author  = {Mead, Lawrence R. and Papanicolaou, N.},
  title   = {Maximum entropy in the problem of moments},
  journal = {Journal of Mathematical Physics},
  volume  = {25},
  number  = {8},
  pages   = {2404--2417},
  year    = {1984},
  doi     = {10.1063/1.526446}
}

@article{borwein1991convergence,
  author  = {Borwein, Jonathan M. and Lewis, Adrian S.},
  title   = {Convergence of best entropy estimates},
  journal = {SIAM Journal on Optimization},
  volume  = {1},
  number  = {2},
  pages   = {191--205},
  year    = {1991},
  doi     = {10.1137/0801014}
}

@article{bandyopadhyay2005stable,
  author  = {Bandyopadhyay, Kausik and Bhattacharya, Arun K. and Biswas, Parthapratim and Drabold, David A.},
  title   = {Maximum entropy and the problem of moments: A stable algorithm},
  journal = {Physical Review E},
  volume  = {71},
  number  = {5},
  pages   = {057701},
  year    = {2005},
  doi     = {10.1103/PhysRevE.71.057701}
}

@article{abramov2007improved,
  author  = {Abramov, Rafail V.},
  title   = {An improved algorithm for the multidimensional moment-constrained maximum entropy problem},
  journal = {Journal of Computational Physics},
  volume  = {226},
  number  = {1},
  pages   = {621--644},
  year    = {2007},
  doi     = {10.1016/j.jcp.2007.04.026}
}

@article{noviinverardi2003fractional,
  author  = {Novi Inverardi, Pier Luigi and Tagliani, Aldo},
  title   = {Maximum entropy density estimation from fractional moments},
  journal = {Communications in Statistics --- Theory and Methods},
  volume  = {32},
  number  = {2},
  pages   = {327--345},
  year    = {2003},
  doi     = {10.1081/STA-120018189}
}

@article{gzyl2010hausdorff,
  author  = {Gzyl, Henryk and Tagliani, Aldo},
  title   = {Hausdorff moment problem and fractional moments},
  journal = {Applied Mathematics and Computation},
  volume  = {216},
  number  = {11},
  pages   = {3319--3328},
  year    = {2010},
  doi     = {10.1016/j.amc.2010.04.059}
}

@article{cottone2009fractional,
  author  = {Cottone, Giulio and Di Paola, Mario},
  title   = {On the use of fractional calculus for the probabilistic characterization of random variables},
  journal = {Probabilistic Engineering Mechanics},
  volume  = {24},
  number  = {3},
  pages   = {321--330},
  year    = {2009},
  doi     = {10.1016/j.probengmech.2008.08.002}
}

@article{li2011combined,
  author  = {Li, Gang and Zhang, Kai},
  title   = {A combined reliability analysis approach with dimension reduction method and maximum entropy method},
  journal = {Structural and Multidisciplinary Optimization},
  volume  = {43},
  number  = {1},
  pages   = {121--134},
  year    = {2011},
  doi     = {10.1007/s00158-010-0546-2}
}

@article{zhang2013entropy,
  author  = {Zhang, Xufang and Pandey, Mahesh D.},
  title   = {Structural reliability analysis based on the concepts of entropy, fractional moment and dimensional reduction method},
  journal = {Structural Safety},
  volume  = {43},
  pages   = {28--40},
  year    = {2013},
  doi     = {10.1016/j.strusafe.2013.03.001}
}

@article{li2019improved,
  author  = {Li, Gang and He, Wanxin and Zeng, Yan},
  title   = {An improved maximum entropy method via fractional moments with {L}aplace transform for reliability analysis},
  journal = {Structural and Multidisciplinary Optimization},
  volume  = {59},
  number  = {4},
  pages   = {1301--1320},
  year    = {2019},
  doi     = {10.1007/s00158-018-2129-6}
}

@article{xu2019novel,
  author  = {Xu, Jun and Dang, Chao},
  title   = {A novel fractional moments-based maximum entropy method for high-dimensional reliability analysis},
  journal = {Applied Mathematical Modelling},
  volume  = {75},
  pages   = {749--768},
  year    = {2019},
  doi     = {10.1016/j.apm.2019.06.037}
}

@article{zhang2020fractional,
  author  = {Zhang, Xiaodong and Low, Ying Min and Koh, Chan Ghee},
  title   = {Maximum entropy distribution with fractional moments for reliability analysis},
  journal = {Structural Safety},
  volume  = {83},
  pages   = {101904},
  year    = {2020},
  doi     = {10.1016/j.strusafe.2019.101904}
}

@article{li2022multimodal,
  author  = {Li, Gang and Wang, Yixuan and Zeng, Yan and He, Wanxin},
  title   = {A new maximum entropy method for estimation of multimodal probability density function},
  journal = {Applied Mathematical Modelling},
  volume  = {102},
  pages   = {137--152},
  year    = {2022},
  doi     = {10.1016/j.apm.2021.09.029}
}

@article{li2024improved,
  author  = {Li, Gang and Wang, Yixuan and Zeng, Yan and He, Wanxin},
  title   = {An improved fractional moment maximum entropy method with polynomial fitting},
  journal = {Journal of Mechanical Design},
  volume  = {146},
  number  = {6},
  pages   = {061704},
  year    = {2024},
  doi     = {10.1115/1.4064247}
}

@article{wang2025iterative,
  author  = {Wang, Lei and Wang, Tao and Dong, You and Frangopol, Dan M. and Li, Zhengliang},
  title   = {Structural reliability analysis based on fractional moments-based iterative maximum entropy method and multiplicative exact dimension reduction integration method},
  journal = {Reliability Engineering \& System Safety},
  volume  = {264},
  pages   = {111344},
  year    = {2025},
  doi     = {10.1016/j.ress.2025.111344}
}

@phdthesis{burg1975mesa,
  author = {Burg, John Parker},
  title  = {Maximum Entropy Spectral Analysis},
  school = {Department of Geophysics, Stanford University},
  year   = {1975}
}

@article{lang1982mem,
  author  = {Lang, Stephen W. and McClellan, James H.},
  title   = {Multidimensional {MEM} spectral estimation},
  journal = {IEEE Transactions on Acoustics, Speech, and Signal Processing},
  volume  = {30},
  number  = {6},
  pages   = {880--887},
  year    = {1982},
  doi     = {10.1109/TASSP.1982.1163967}
}

@article{gull1978image,
  author  = {Gull, S. F. and Daniell, G. J.},
  title   = {Image reconstruction from incomplete and noisy data},
  journal = {Nature},
  volume  = {272},
  pages   = {686--690},
  year    = {1978},
  doi     = {10.1038/272686a0}
}

@article{barron1991exponential,
  author  = {Barron, Andrew R. and Sheu, Chyong-Hwa},
  title   = {Approximation of density functions by sequences of exponential families},
  journal = {The Annals of Statistics},
  volume  = {19},
  number  = {3},
  pages   = {1347--1369},
  year    = {1991},
  doi     = {10.1214/aos/1176348252}
}

@article{gatto2007gvm,
  author  = {Gatto, Riccardo and Jammalamadaka, Sreenivasa Rao},
  title   = {The generalized von {M}ises distribution},
  journal = {Statistical Methodology},
  volume  = {4},
  number  = {3},
  pages   = {341--353},
  year    = {2007},
  doi     = {10.1016/j.stamet.2006.11.003}
}

@article{feuerverger1977ecf,
  author  = {Feuerverger, Andrey and Mureika, Roman A.},
  title   = {The empirical characteristic function and its applications},
  journal = {The Annals of Statistics},
  volume  = {5},
  number  = {1},
  pages   = {88--97},
  year    = {1977},
  doi     = {10.1214/aos/1176343742}
}

@article{johnson1949systems,
  author  = {Johnson, N. L.},
  title   = {Systems of frequency curves generated by methods of translation},
  journal = {Biometrika},
  volume  = {36},
  number  = {1-2},
  pages   = {149--176},
  year    = {1949},
  doi     = {10.1093/biomet/36.1-2.149}
}

@article{daniels1954saddlepoint,
  author  = {Daniels, H. E.},
  title   = {Saddlepoint approximations in statistics},
  journal = {The Annals of Mathematical Statistics},
  volume  = {25},
  number  = {4},
  pages   = {631--650},
  year    = {1954},
  doi     = {10.1214/aoms/1177728652}
}

@article{blinnikov1998expansions,
  author  = {Blinnikov, S. and Moessner, R.},
  title   = {Expansions for nearly {G}aussian distributions},
  journal = {Astronomy and Astrophysics Supplement Series},
  volume  = {130},
  number  = {1},
  pages   = {193--205},
  year    = {1998},
  doi     = {10.1051/aas:1998221}
}

@article{provost2005moment,
  author  = {Provost, Serge B.},
  title   = {Moment-based density approximants},
  journal = {Mathematica Journal},
  volume  = {9},
  number  = {4},
  pages   = {727--756},
  year    = {2005}
}

@techreport{jcgm100,
  author      = {{Joint Committee for Guides in Metrology}},
  title       = {Evaluation of measurement data --- Guide to the expression of uncertainty in measurement},
  number      = {JCGM 100:2008},
  institution = {BIPM, IEC, IFCC, ILAC, ISO, IUPAC, IUPAP and OIML},
  year        = {2008},
  note        = {GUM 1995 with minor corrections},
  doi         = {10.59161/JCGM100-2008E}
}

@techreport{jcgm101,
  author      = {{Joint Committee for Guides in Metrology}},
  title       = {Evaluation of measurement data --- Supplement 1 to the ``Guide to the expression of uncertainty in measurement'' --- Propagation of distributions using a {M}onte {C}arlo method},
  number      = {JCGM 101:2008},
  institution = {BIPM, IEC, IFCC, ILAC, ISO, IUPAC, IUPAP and OIML},
  year        = {2008},
  doi         = {10.59161/JCGM101-2008}
}

@article{rajan2016benchmark,
  author  = {Rajan, Arvind and Kuang, Ye Chow and Ooi, Melanie Po-Leen and Demidenko, Serge N.},
  title   = {Benchmark test distributions for expanded uncertainty evaluation algorithms},
  journal = {IEEE Transactions on Instrumentation and Measurement},
  volume  = {65},
  number  = {5},
  pages   = {1022--1034},
  year    = {2016},
  doi     = {10.1109/TIM.2015.2507418}
}

@article{rajan2018moment,
  author  = {Rajan, Arvind and Kuang, Ye Chow and Ooi, Melanie Po-Leen and Demidenko, Serge N. and Carstens, Herman},
  title   = {Moment-constrained maximum entropy method for expanded uncertainty evaluation},
  journal = {IEEE Access},
  volume  = {6},
  pages   = {4072--4082},
  year    = {2018},
  doi     = {10.1109/ACCESS.2017.2787736}
}

@article{morello2020gum,
  author  = {Morello, Rosario},
  title   = {{GUM}-based decisional criteria to make decisions in presence of measurement uncertainty},
  journal = {IEEE Transactions on Instrumentation and Measurement},
  volume  = {69},
  number  = {8},
  pages   = {5511--5522},
  year    = {2020},
  doi     = {10.1109/TIM.2019.2963581}
}

@book{kunchenko2002polynomial,
  title     = {Polynomial Parameter Estimations of Close to {G}aussian Random Variables},
  author    = {Kunchenko, Yuriy P.},
  publisher = {Shaker Verlag},
  address   = {Aachen},
  year      = {2002},
  note      = {in English}
}

@book{kunchenko2005approximation,
  title     = {Polynomials of Approximation in a Space with a Generating Element},
  author    = {Kunchenko, Yuriy P.},
  publisher = {Naukova Dumka},
  address   = {Kyiv},
  year      = {2005},
  note      = {in Ukrainian}
}

@book{kunchenko2006stochastic,
  title     = {Stochastic Polynomials},
  author    = {Kunchenko, Yuriy P.},
  publisher = {Naukova Dumka},
  address   = {Kyiv},
  year      = {2006},
  pages     = {275},
  note      = {in Russian}
}

@article{zabolotnii2026parametrically,
  title   = {Parametrically adaptive transition polynomial: a signed-parity continuous-alpha extension of {K}unchenko stochastic polynomials},
  author  = {Zabolotnii, Serhii},
  journal = {arXiv preprint arXiv:2605.14610},
  year    = {2026}
}

@article{zabolotnii2026variance,
  title   = {Variance-reduced manifold sampling via polynomial-maximization density estimation},
  author  = {Zabolotnii, Serhii},
  journal = {arXiv preprint arXiv:2605.19938},
  year    = {2026}
}

@article{csiszar1975divergence,
  title   = {{$I$}-divergence geometry of probability distributions and minimization problems},
  author  = {Csisz{\'a}r, Imre},
  journal = {The Annals of Probability},
  volume  = {3},
  number  = {1},
  pages   = {146--158},
  year    = {1975}
}

@book{golan1996maximum,
  title     = {Maximum Entropy Econometrics: Robust Estimation with Limited Data},
  author    = {Golan, Amos and Judge, George and Miller, Douglas},
  publisher = {John Wiley \& Sons},
  address   = {Chichester},
  year      = {1996}
}

@article{tsallis1988possible,
  title   = {Possible generalization of {B}oltzmann--{G}ibbs statistics},
  author  = {Tsallis, Constantino},
  journal = {Journal of Statistical Physics},
  volume  = {52},
  number  = {1--2},
  pages   = {479--487},
  year    = {1988},
  doi     = {10.1007/BF01016429}
}

@article{jones1993simple,
  title   = {Simple boundary correction for kernel density estimation},
  author  = {Jones, M. C.},
  journal = {Statistics and Computing},
  volume  = {3},
  number  = {3},
  pages   = {135--146},
  year    = {1993},
  doi     = {10.1007/BF00147776}
}

@article{efromovich2010orthogonal,
  title   = {Orthogonal series density estimation},
  author  = {Efromovich, Sam},
  journal = {Wiley Interdisciplinary Reviews: Computational Statistics},
  volume  = {2},
  number  = {4},
  pages   = {467--476},
  year    = {2010}
}

@book{shohat1943problem,
  author    = {Shohat, James A. and Tamarkin, Jacob D.},
  title     = {The Problem of Moments},
  series    = {Mathematical Surveys},
  volume    = {1},
  publisher = {American Mathematical Society},
  address   = {Providence, RI},
  year      = {1943}
}

@book{akhiezer1965classical,
  author    = {Akhiezer, Naum I.},
  title     = {The Classical Moment Problem and Some Related Questions in Analysis},
  publisher = {Oliver \& Boyd},
  address   = {Edinburgh},
  year      = {1965}
}

@book{barndorff1978information,
  author    = {Barndorff-Nielsen, Ole},
  title     = {Information and Exponential Families in Statistical Theory},
  publisher = {Wiley},
  address   = {Chichester},
  year      = {1978}
}

@book{kapur1989maximum,
  author    = {Kapur, Jagat Narain},
  title     = {Maximum-Entropy Models in Science and Engineering},
  publisher = {Wiley Eastern},
  address   = {New Delhi},
  year      = {1989}
}

@article{zellner1988calculation,
  author  = {Zellner, Arnold and Highfield, Richard A.},
  title   = {Calculation of maximum entropy distributions and approximation of marginal posterior distributions},
  journal = {Journal of Econometrics},
  volume  = {37},
  number  = {2},
  pages   = {195--209},
  year    = {1988},
  doi     = {10.1016/0304-4076(88)90002-4}
}

@article{wu2003calculation,
  author  = {Wu, Ximing},
  title   = {Calculation of maximum entropy densities with application to income distribution},
  journal = {Journal of Econometrics},
  volume  = {115},
  number  = {2},
  pages   = {347--354},
  year    = {2003},
  doi     = {10.1016/S0304-4076(03)00114-3}
}

@article{rockinger2002entropy,
  author  = {Rockinger, Michael and Jondeau, Eric},
  title   = {Entropy densities with an application to autoregressive conditional skewness and kurtosis},
  journal = {Journal of Econometrics},
  volume  = {106},
  number  = {1},
  pages   = {119--142},
  year    = {2002},
  doi     = {10.1016/S0304-4076(01)00092-6}
}

@article{aleagobe2023length,
  author  = {Aleagobe, Hussain A. A. and Fakoor, Vahid and Jomhoori, Sarah},
  title   = {Maximum entropy and empirical likelihood in length-biased setting: a comparison study},
  journal = {Communications in Statistics---Simulation and Computation},
  year    = {2023},
  volume  = {52},
  number  = {6},
  pages   = {2439--2452},
  doi     = {10.1080/03610918.2021.1908555}
}

@article{tohme2024messy,
  author  = {Tohme, Tony and Sadr, Mohsen and Youcef-Toumi, Kamal and Hadjiconstantinou, Nicolas G.},
  title   = {{MESSY} estimation: maximum-entropy based stochastic and symbolic density estimation},
  journal = {Transactions on Machine Learning Research},
  year    = {2024},
  url     = {https://openreview.net/forum?id=Y2ru0LuQeS}
}

@article{soize2023fourier,
  author  = {Soize, Christian},
  title   = {Probabilistic learning constrained by realizations using a weak formulation of {Fourier} transform of probability measures},
  journal = {Computational Statistics},
  year    = {2023},
  volume  = {38},
  number  = {4},
  pages   = {1879--1925},
  doi     = {10.1007/s00180-022-01300-w}
}

@article{nghiem2026relaxed,
  author  = {Nghiem, Thi Lich and Mar{\'e}chal, Pierre},
  title   = {On maximum entropy density estimation with relaxed moment constraints},
  journal = {Entropy},
  year    = {2026},
  volume  = {28},
  number  = {3},
  pages   = {282},
  doi     = {10.3390/e28030282}
}

\end{document}